\documentclass[11pt,a4paper]{article}

\usepackage{amssymb}
\usepackage{amsmath}
\usepackage{graphicx}
\usepackage{color}
\usepackage{enumerate}
\usepackage{mathrsfs}

\graphicspath{{./}{figures/}}

\newtheorem{theorem}{Theorem}[section]
\newtheorem{lemma}{Lemma}[section]

\newbox\ProofSym
\setbox\ProofSym=\hbox{%
	\unitlength=0.18ex%
	\begin{picture}(10,10)
		\put(0,0){\framebox(9,9){}}
		\put(0,3){\framebox(6,6){}}
\end{picture}}

\newcommand{\eps}{\varepsilon}
\newcommand{\pr}[1]{\mathrm{Pr}\bigl[#1\bigr]}

\newcommand{\cancel}[1]{}

\newcommand{\E}[1] {\mathrm{E}\bigl[#1\bigr]}
\newcommand{\EE}[1]{\mathrm{E}\left[#1\right]}

\newcommand{\del}{\mathrm{Del}_Q}
\newcommand{\vor}{\mathrm{Vor}_Q}
\newcommand{\NN}{\mathrm{NN}}

\newcommand{\Int}{\mathrm{int}}

\begin{document}

	
\title{Self-Improving Voronoi Construction for a Hidden Mixture of Product Distributions\thanks{Research of Cheng and Wong are supported by Research Grants Council, Hong Kong, China (project no.~16200317).}}

\author{Siu-Wing Cheng\footnote{Department of Computer Science and Engineering, HKUST, Hong Kong, China.  Email: {\tt schenng@cse.ust.hk, mtwongaf@connect.ust.hk}} \quad \quad Man Ting Wong\footnotemark[2]}

\date{}
\maketitle

\begin{abstract}
	We propose a self-improving algorithm for computing Voronoi diagrams under a given convex distance function with constant description complexity.  The $n$ input points are drawn from a hidden mixture of product distributions; we are only given an upper bound $m = o(\sqrt{n})$ on the number of distributions in the mixture, and the property that for each distribution, an input instance is drawn from it with a probability of $\Omega(1/n)$.   For any $\varepsilon \in (0,1)$, after spending $O\bigl(mn\log^{O(1)} (mn) + m^{\varepsilon} n^{1+\varepsilon}\log(mn)\bigr)$ time in a training phase, our algorithm achieves an $O\bigl(\frac{1}{\varepsilon}n\log m + \frac{1}{\varepsilon}n2^{O(\log^* n)} + \frac{1}{\varepsilon}H\bigr)$ expected running time with probability at least $1 - O(1/n)$, where $H$ is the entropy of the distribution of the Voronoi diagram output.  The expectation is taken over the input distribution and the randomized decisions of the algorithm.  For the Euclidean metric, the expected running time improves to $O\bigl(\frac{1}{\varepsilon}n\log m +  \frac{1}{\varepsilon}H\bigr)$.  
\end{abstract}


\section{Introduction}
\label{sect:introduction}

Self-improving algorithms, proposed by Ailon et al.~\cite{ailon11}, is a framework for studying algorithmic complexity beyond the worst case.  There is a \emph{training phase} that allows some auxiliary structures about the input distribution to be constructed.  In the \emph{operation phase}, these auxiliary structures help to achieve an expected running time, called the \emph{limiting complexity}, that may surpass the worst-case optimal time complexity.

Self-improving algorithms have been designed for \emph{product distributions}~\cite{ailon11,clarkson14}.  Let $n$ be the input size.  A product distribution $\mathscr{D} = (D_1,\ldots,D_n)$ consists of $n$ distributions $D_i$ such that the $i$th input item is drawn independently from $D_i$.  It is possible that $D_i = D_j$ for some $i \not= j$, but the draws of the $i$th and $j$th input items are independent.  No further information about $\mathscr{D}$ is given.  Sorting, Delaunay triangulation, 2D maxima, and 2D convex hull have been studied for product distributions.  For all four problems, the training phase uses $O(n^{\eps})$ input instances, and the space complexity is $O(n^{1+\eps})$.  The limiting complexities of sorting and Delaunay triangulation are $O\bigl(\frac{1}{\eps}n + \frac{1}{\eps}H_{\text{out}}\bigr)$ for any $\eps \in (0,1)$, where $H_{\text{out}}$ is the entropy of the output distribution~\cite{ailon11}.  The limiting complexities for 2D maxima and 2D convex hull are  $O(\mathrm{OptM} + n)$ and $O(\mathrm{OptC} + n\log\log n)$ respectively, where OptM and OptC are the expected depths of the optimal linear decision trees for the two problems~\cite{clarkson14}.   

Extensions that allow dependence among input items have been developed.  One extension is that there is a \emph{hidden partition} of $[n]$ into groups.  The input items with indices in the $k$th group follow some \emph{hidden functions} of a common parameter $u_k$.  The parameters $u_1, u_2, \cdots$ follow a product distribution.  The partition of $[n]$ is not given though.  If the hidden functions are known to be linear, sorting can be solved in a limiting complexity of $O\bigl(\frac{1}{\eps}n + \frac{1}{\eps}H_{\text{out}}\bigr)$ after a training phase that takes $O(n^2\log^3 n)$ time~\cite{cheng20b}.  If it is only known that each hidden function has $O(1)$ extrema and the graphs of two functions intersect in $O(1)$ places (without knowing any of the functions, or any of these extrema and intersections), sorting can be solved in a limiting complexity of $O(n + H_{\text{out}})$ after an $\tilde{O}(n^3)$-time training phase~\cite{cheng20a}.    For the Delaunay triangulation problem, if it is known that the hidden functions are bivariate polynomials of $O(1)$ degree (without knowing the polynomials), a limiting complexity of $O(n\alpha(n) + H_{\text{out}})$ can be achieved after a polynomial-time training phase~\cite{cheng20a}.  

Another extension is that the input instance $I$ is drawn from a \emph{hidden mixture} of at most $m$ product distributions.  That is, there are at most $m$ product distributions $\mathscr{D}_1, \mathscr{D}_2, \ldots$ such that $\Pr[I \sim \mathscr{D}_a] = \lambda_a$ for some fixed positive value $\lambda_a$.  The upper bound $m$ is given, but no information about the $\lambda_a$'s and the $\mathscr{D}_a$'s is provided.  Sorting can be solved in a limiting complexity of $O\bigl(\frac{1}{\eps}n\log m + \frac{1}{\eps}H_{\text{out}}\bigr)$ after a training phase that takes $O(mn\log^2 (mn) + m^\eps n^{1+\eps}\log(mn))$ time~\cite{cheng20b}.

In this paper, we present a self-improving algorithm for constructing Voronoi diagrams under a convex distance function $d_Q$ in $\mathbb{R}^2$, assuming that the input distribution is a hidden mixture of at most $m$ product distributions.  The convex distance function $d_Q$ is induced by a given convex polygon $Q$ of $O(1)$ size.  The upper bound $m$ is given, and we assume that $m = o(\sqrt{n})$.  We also assume that for each product distribution $\mathscr{D}_a$ in the mixture, $\lambda_a = \Omega(1/n)$.   Let $\eps \in (0,1)$ be a parameter fixed beforehand.  The training phase uses $O(mn\log(mn))$ input instances and takes $O\bigl(mn\log^{O(1)} (mn) + m^{\varepsilon} n^{1+\varepsilon}\log(mn)\bigr)$ time.   In the operation phase, given an input instance $I$, we can construct its Voronoi diagram $\vor(I)$ under $d_Q$ in a limiting complexity of $O\big(\frac{1}{\eps}n\log m + \frac{1}{\eps}n2^{O(\log^* n)} + \frac{1}{\eps}H\bigr)$, where $H$ denotes the entropy of the distribution of the Voronoi diagram output.   
Note that $\Omega(H)$ is a lower bound of the expected running time of any comparison-based algorithm. Our algorithm also works for the Euclidean case, and the limiting complexity improves to $O\big(\frac{1}{\eps}n\log m + \frac{1}{\eps}H\bigr)$.  

For simplicity, we will assume throughout the rest of this paper that the hidden mixture has exactly $m$ product distributions.  We give an overview of our method in the following.

We follow the strategy in~\cite{ailon11} for computing a Euclidean Delaunay triangulation.  The idea is to form a set $S$ of sample points and build $\mathrm{Del}(S)$ and some auxiliary structures in the training phase so that any future input instance $I$ can be merged quickly into $\mathrm{Del}(S)$ to form $\mathrm{Del}(S \cup I)$, and then $\mathrm{Del}(I)$ can be split off in $O(n)$ expected time.  Merging $I$ into $\mathrm{Del}(S)$ requires locating the input points in $\mathrm{Del}(S)$.  The location distribution is gathered in the training phase so that distribution-sensitive point location can be used to avoid the logarithmic query time as much as possible.  Modifying $\mathrm{Del}(S)$ efficiently into $\mathrm{Del}(S \cup I)$ requires that only $O(1)$ points in $I$ fall into the same neighborhood in $\mathrm{Del}(S)$ in expectation.

In our case, since there are $m$ product distributions, we will need a larger set $S$ of $mn$ sample points in order to ensure that only $O(1)$ points in $I$ fall into the same neighborhood in $\vor(S)$ in expectation.  But then merging $I$ into $\vor(S)$ in the operation phase would be too slow because scanning $\vor(S)$ already requires $\Theta(mn)$ time.  We need to extract a subset $R \subseteq S$ such that $R$ has $O(n)$ size and $R$ contains all points in $S$ whose Voronoi cells conflict with the input points.

Still, we cannot afford to construct $\vor(R)$ in $O(n\log n)$ time.  In the training phase, we form a metric $d$ related to $d_Q$ and construct a \emph{net-tree} $T_S$ for $S$ under $d$~\cite{har-peled06}.  In the operation phase, after finding the appropriate $R \subseteq S$, we use nearest common ancestor queries~\cite{TL88} to compress $T_S$ in $O(n\log\log m)$ time to a subtree $T_R$ for $R$ that has $O(n)$ size.  Next, we use $T_R$ to construct a well-separated pair decomposition of $R$ under $d$ in $O(n)$ time~\cite{har-peled06}, use the decomposition to compute the nearest neighbor graph of $R$ under $d$ in $O(n)$ time, and then construct $\vor(R)$ from the nearest neighbor graph in $O(n)$ expected time.    The merging of $I$ into $\vor(R)$ to form $\vor(R \cup I)$, and the splitting of $\vor(R \cup I)$ into $\vor(I)$ and $\vor(R)$ are obtained by transferring their analogous results in the Euclidean case~\cite{ailon11,chazelle02}.  

We have left out the expected time to locate the input points in $\vor(S)$.  It is bounded by $O(1/\varepsilon)$ times the sum of the entropies of the point location outcomes.  We show that $\vor(I)$ allows us to locate the input points in $\vor(S)$ in $O(n\log m  + n2^{O(\log^* n)})$ time. Then, a result in~\cite{ailon11} implies that the sum of the entropies of the point location outcomes is $O(n\log m + n2^{O(\log^* n)} + H)$.   The expected running time is thus $O(\frac{1}{\eps}n\log m + \frac{1}{\eps}n2^{O(\log^* n)} + \frac{1}{\eps}H)$, which dominates the limiting complexity.  In the Euclidean case, $\mathrm{Vor}(I)$ allows us to locate the input points in $O(n\log m)$ time, so the limiting complexity improves to $O(\frac{1}{\eps}n\log m + \frac{1}{\eps}H)$.

\section{Preliminaries}

Let $Q$ be a convex polygon that has $O(1)$ complexity and contains the origin in its interior.  Let $\partial$ and $\Int(\cdot)$ be the boundary and interior operators, respectively.  So $Q$'s boundary is $\partial Q$ and its interior is $\Int(Q)$.  Let $d_Q$ be the distance function induced by $Q$: $\forall \, x, y \in \mathbb{R}^2$, $d_Q(x,y) = \min\{\lambda \in [0,\infty) : y \in \lambda Q + x\}$.  As $Q$ may not be centrally symmetric (i.e., $x \in Q \iff -x \in Q$), $d_Q$ may not be a metric.

The bisector of two points $p$ and $q$ is $\{x \in \mathbb{R}^2 : d_Q(p,x) = d_Q(q,x)\}$, which is an open polygonal curve of $O(1)$ size.  The Voronoi diagram of a set $\Sigma$ of $n$ points, $\vor(\Sigma)$, is a partition of $\mathbb{R}^2$ into interior-disjoint cells $V_p(\Sigma) = \{x \in \mathbb{R}^2: \, \forall q \in \Sigma, \, d_Q(p,x) \leq d_Q(q,x)\}$ for all $p \in \Sigma$.   There are algorithms for constructing $\vor(\Sigma)$ in $O(n\log n)$ time~\cite{CD85,KMM93}.  

$V_p(\Sigma)$ is simply connected and star-shaped with respect to $p$~\cite{CD85}.  We use $N_p(\Sigma)$ to denote the set of Voronoi neighbors of $p$ in $\vor(\Sigma)$.  The Voronoi edges of $\vor(\Sigma)$ form a planar graph of $O(|\Sigma|)$ size.  Each Voronoi edge is a polygonal line, and we call its internal vertices \emph{Voronoi edge bends}.   We use $V_\Sigma$ to denote the set of Voronoi edge bends and Voronoi vertices in $\vor(\Sigma)$.  For the infinite Voronoi edges, their endpoints at infinity are included in $V_\Sigma$.


Define $Q^* = \{-x : x \in Q\}$.  For any points $x,y  \in \mathbb{R}^2$, $d_{Q^*}(x,y) = d_Q(y,x)$.  At any point $x$ on a Voronoi edge of $\vor(\Sigma)$ defined by $p,q \in \Sigma$, there exists $\lambda \in (0,\infty)$ such that $d_{Q^*}(x,p) = d_Q(p,x) = d_Q(q,x) = d_{Q^*}(x,q) = \lambda$ and $d_{Q^*}(x,s) = d_Q(s,x) \geq \lambda$ for all $s \in \Sigma$.  Hence, $\{p,q\} \subset \partial(\lambda Q^* + x)$  and $\Int(\lambda Q^* + x) \cap \Sigma = \emptyset$, i.e., an ``empty circle property''.

Take a point $x$.  Consider the largest homothetic\footnote{A homothetic copy of a shape is a scaled and translated copy of it.} copy $Q^*_x$ of $Q^*$ centered at $x$ such that $\Int(Q^*_x) \cap \Sigma = \emptyset$.  If we insert a new point $q$ to $\Sigma$, we say that $q$ \emph{conflicts with} $x$ if $q \in Q^*_x$.  We say that $q$ \emph{conflicts with} a cell $V_p(\Sigma)$ if $q$ conflicts with some point in $V_p(\Sigma)$.  Clearly, $V_p(\Sigma)$ must be updated by the insertion of $q$.  We use $V_{\Sigma}|_q$ to denote the subset of $V_\Sigma$ that conflict with $q$.   The Voronoi edge bends and Voronoi vertices in $V_\Sigma|_q$ will be destroyed by the insertion of $q$.

We make three general position assumptions.  First, no two sides of $Q$ are parallel.  Second, for every pair of input points, their support line is not parallel to any side of $Q$.  Third, no four input points lie on the boundary of any homothetic copy of $Q^*$, which implies that  every Voronoi vertex  has degree three.  

It is much more convenient if all Voronoi cells of the input points are bounded.  We assume that all possible input points appear in some fixed bounding square $\cal B$ centered at the origin.  We place $O(1)$ dummy points outside $\cal B$ so that all Voronoi cells of the input points are bounded, and their portions inside $\cal B$ remain the same as before.  Refer to Figure~\ref{fg:dummy}.  Take $\lambda Q^*$ for some large enough $\lambda \in \mathbb{R}$ such that for every point $x \in {\cal B}$, $\lambda Q^* + x$ contains $\cal B$.  Refer to the left image in Figure~\ref{fg:dummy}.   We slide a copy of $\lambda Q^*$ around $\cal B$ to generate the outer convex polygon.  The dashed polygon demonstrates the sliding of $\lambda Q^*$ around $\cal B$.  This outer polygon contains a translational copy of every edge of $\cal B$ and two translational copies of every edge of $\lambda Q^*$.  We add the vertices of this outer polygon as dummy points.  Any homothetic copy of $Q^*$ that intersects $\cal B$ cannot be expanded indefinitely without containing some of these dummy points.  So all Voronoi cells of input points are bounded.  For each point $x \in {\cal B}$, since the dummy points lie outside $\lambda Q^*+x $ and ${\cal B} \subseteq \lambda Q^* + x$ (i.e., $\lambda Q^* +x$ is not empty of the input points), the portion of the Voronoi diagram inside $\cal B$ is unaffected by the dummy points.

\begin{figure}
	\centerline{\includegraphics[scale=0.7]{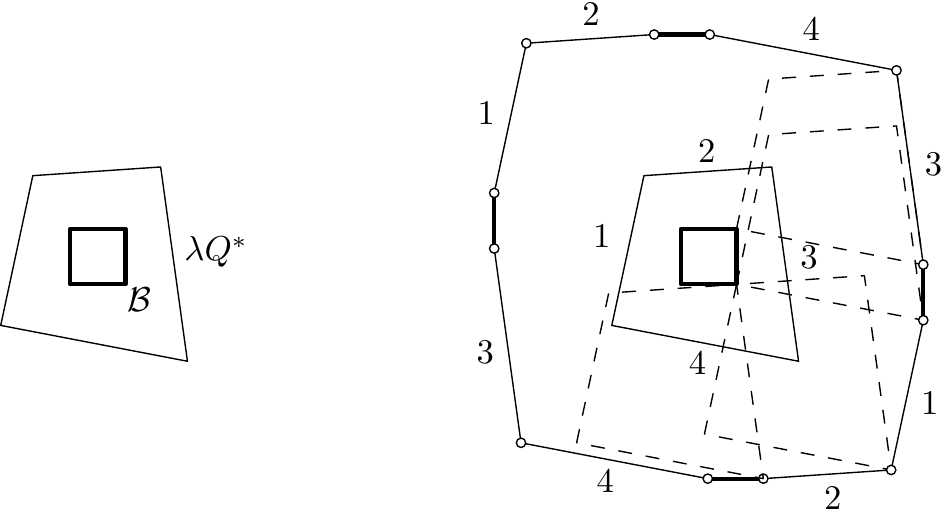}}
	\caption{The left image shows the bounding square $\cal B$ and the large enclosing $\lambda Q^*$.  In the right image, we slide a copy of $\lambda Q^*$ around $\cal B$ to generate the outer convex polygon.  The dashed polygon demonstrates the sliding of $\lambda Q^*$ around $\cal B$.  The bold edges on this convex polygon are translates of the boundary edges of $\cal B$.   Every edge of $\lambda Q^*$ has two translational copies too as labelled.}
	\label{fg:dummy}
\end{figure}

\section{Training phase}
\label{sec:train}

\noindent {\bf Sample set $\pmb S$.}  Take $mn\ln (mn)$ instances $I_1, I_2, \ldots, I_{mn\ln(mn)}$.  Define $x_1,\ldots, x_{mn\ln(mn)}$ by taking the $p_1$'s in $I_1,\ldots,I_{m\ln(mn)}$ to be $x_1,\ldots,x_{m\ln(mn)}$,  $p_2$'s in $I_{m\ln(mn)+1},\ldots,I_{2m\ln(mn)}$ to be $x_{m\ln(mn)+1},\ldots,x_{2m\ln(mn)}$, and so on.   The set $S$ of sample points includes a $\frac{1}{mn}$-net of the $x_i$'s with respect to the family of homothetic copies of $Q^*$, as well as the $O(1)$ dummy points.  The set $S$ has $O(mn)$ points and can be constructed in $O(mn\log^{O(1)} (mn))$ time as homothetic copies of $Q^*$ are pseudo-disks~\cite{ailon11,PR08}.  

\vspace{6pt}

\noindent {\bf Point location.}  Compute $\vor(S)$ and triangulate it by connecting each $p \in S$ to $V_S \cap \partial V_p(S)$, i.e., the Voronoi edge bends and Voronoi vertices in $\partial V_p(S)$.   For unbounded Voronoi cells, we view the infinite Voronoi edges as leading to some vertices at infinity; an extra triangulation edge that goes between two infinite Voronoi edges also leads to a vertex at infinity, giving rise to unbounded triangles.  
Figure~\ref{fg:example} shows an example.

\begin{figure}
	\centerline{\includegraphics[scale=0.425]{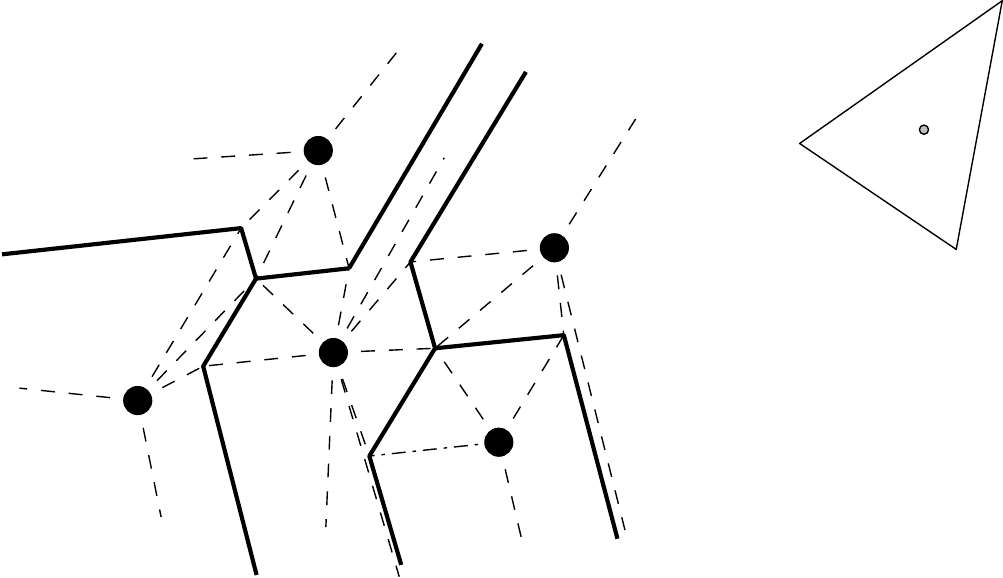}}
	\caption{Part of the triangulation of a Voronoi diagram induced by the triangle shown with a gray center.  The solid edges form the Voronoi diagram.  The dashed edges refine it into a triangulation.}
	\label{fg:example}
\end{figure}

Construct a point location structure $L_S$ for the triangulated $\vor(S)$ with $O(\log (mn))$ query time~\cite{edels86}.  Take another $m^\eps n^\eps$ input instances and use $L_S$ to locate the points in these input instances in the triangulated $\vor(S)$.  For every $i \in [n]$ and every triangle $t$, we compute $\tilde{\pi}_{i,t}$ to be the ratio of the frequency of $t$ hit by $p_i$ to $m^\eps n^\eps$, which is an estimate of $\Pr[p_i \in t]$.  For each $i \in [n]$, form a subdivision ${\cal S}_i$ that consists of triangles with positive $\tilde{\pi}_{i,t}$'s, triangulate the exterior of ${\cal S}_i$, and give these new triangles a zero estimated probability.  Set the weight of each triangle in ${\cal S}_i$ to be the maximum of $(mn)^{-\eps}$ and its estimated probability.  Construct a distribution-sensitive point location structure $L_i$ for ${\cal S}_i$ based on the triangle weights~\cite{arya07,iacono04}.  Note that $L_i$ has $O(m^\eps n^\eps)$ size, and locating a point in a triangle $t \in {\cal S}_i$ takes $O\bigl(\log \frac{W_i}{w_t} \bigr)$ time, where $w_t$ is the weight of $t$ and $W_i$ is the total weight in ${\cal S}_i$.

For any input instance $(p_1,\ldots,p_n)$ in the operation phase, we will query $L_i$ to locate $p_i$ in the triangulated $\vor(S)$, which may fail if $p_i$ falls into a triangle with zero estimated probability.  If the search fails, we query $L_S$ to locate $p_i$.  

\vspace{6pt}

\noindent {\bf Net-tree.}   
%
We first define a metric that is induced by a centrally symmetric convex polygon.   Define $\hat{Q} = \{x - y : x, y \in Q^*\}$, i.e., the Minkowski sum of $Q^*$ and $-Q^*$, or equivalently $Q^*$ and $Q$.  It is centrally symmetric by definition.   It can be visualized as the region covered by all possible placements of $Q^*$ that has the origin in the polygon boundary.  Since $\hat{Q}$ is a Minkowski sum, its number of vertices is within a constant factor of the total number of vertices of $Q^*$ and $-Q^*$, which is $O(1)$.  

Let $d$ be the metric induced by the centrally symmetric convex polygon $\hat{Q}$, which is a \emph{doubling metric}---there is a constant $\lambda > 0$ such that for any point $x \in \mathbb{R}^2$ and any positive number $r$, the ball with respect to $d$ centered at $x$ with radius $r$ can be covered by $\lambda$ balls with respect to $d$ of radius $r/2$.

Given a set of points $P$, a \emph{net-tree} for $P$ with respect to $d$~\cite{har-peled06} is an analog of the well-separated pair decomposition for Euclidean spaces~\cite{CK95}.  It is a rooted tree whose leaves are the points in $P$.  For each node $v$, let $\mathit{parent}(v)$ denote its parent, and let $P_v$ denote the subset of $P$ at the leaves that descend from $v$.  Every tree node $v$ is given a representative point $p_v$ and an integer level $\ell_v$.  Let $\tau \geq 11$ be a fixed constant.  Let $B(x,h)$ denote the ball $\{y \in \mathbb{R}^2 : d(x,y) \leq h\}$.   By the results in~\cite{har-peled06} (Definition~2.1 and the remark that follows Proposition~2.2), the following properties are satisfied by a net-tree:
\begin{enumerate}[\label=(\alph{enumi})]
	\item $p_v \in P_v$.
	\item For every non-root node $v$, $\ell_v < \ell_{\mathit{parent}(v)}$, and if $v$ is a leaf, then $\ell_v = -\infty$.
	\item Every internal node has at least two and at most a constant number of children.
	\item For every node $v$, $B\bigl(p_v,\frac{2\tau}{\tau-1}\cdot\tau^{\ell_v}\bigr)$ contains $P_v$.
	\item For every non-root node $v$, $B\bigl(p_v,\frac{\tau-5}{2\tau-2} \cdot \tau^{\ell_{\mathit{parent}(v)}-1}\bigr) \cap P \subset P_v$.
	\item For every internal node $v$, there is a child $w$ of $v$ such that $p_w = p_v$.
\end{enumerate}

\cancel{
We define a variant of the net-tree that we call the \emph{relaxed net-tree}.  It satisfies the following properties:
\begin{itemize}
	\item It is a rooted tree whose leaves are the points in $P$.
	\item Each node $v$ has a representative point $p_v$, but $p_v$ may not belong to $P_v$.
	\item It satisfies properties~(b)--(e) of a net-tree.
\end{itemize}

Note that a relaxed net-tree may violate property~(f) of a net-tree.
}

\cancel{
The following result shows that for any $k > 0$, the set $\{p_v : \ell_v < k \leq \ell_{p(v)}\}$ is a $\Theta(\tau^k)$-net.  Thus, the net-tree is intuitively a hierarchy of nets with resolution increasing geometrically from the root to the leaves.

\begin{lemma}[\cite{har-peled06}]
	\label{lem:level}
	Let $N_k = \{p_v : \ell_v < k \leq \ell_{p(v)}\}$ for some $k \geq 0$.  For all $x,y \in N_k$, $d(x,y) \geq \tau^{k-1}/4$, and for all $s \in S$, $d(s,N_k) \leq 4\tau^k$.
\end{lemma}

It has been shown that if the underlying metric is doubling, which is true for $d$, then $T_S$ has the following properties.

\begin{lemma}[\cite{har-peled06}]
	\label{lem:net}
	The net-tree $T_S$ for $S$ under the metric $d$ can be constructed in $O(mn\log(mn))$ expected time.  Each internal node of $T_S$ has at least two and at most a constant number of children, and hence $T_S$ has $O(mn)$ nodes.
\end{lemma}
}

\vspace{6pt}

\noindent {\bf Clusters.}  We construct a \emph{net-tree} $T_S$ for $S$ in $O(mn\log(mn))$ expected time~\cite{har-peled06}.  We define \emph{clusters} as follows.  Label all leaves of $T_S$ as unclustered initially.  Select the leftmost $m$ unclustered leaves of $T_S$; if there are fewer than $m$ such leaves, select them all.  Find the subtree rooted at a node $v$ of $T_S$ that contains the selected unclustered leaves, but no child subtree of $v$ contains them all.  We call the subtree rooted at $v$ a cluster and label all its leaves clustered.  Then, we repeat the above until all leaves of $T_S$ are clustered.  By construction, the clusters are disjoint, each cluster has $O(m)$ nodes, and there are $O(n)$ clusters in $T_S$.  

We assign nodes in each cluster a unique cluster index in the range $[1,O(n)]$.  We also assign each node of a cluster three indices from the range $[1,O(m)]$ according to its rank in the preorder, inorder, and postorder traversals of that cluster.  The preorder and postorder indices allow us to tell in $O(1)$ time whether two nodes are an ancestor-descendant pair.  

We keep an initially empty van Emde Boas tree $\mathit{EB}_c$~\cite{boas77} with each cluster $c$.  The universe for $\mathit{EB}_c$ is the set of leaves in the cluster $c$, and the inorder of these leaves in $c$ is the total order for $\mathit{EB}_c$.  We also build a nearest common ancestor query data structure for each cluster~\cite{TL88}.  The nearest common ancestor query of any two nodes can be reported in $O(\log\log m)$ time.

\vspace{6pt}

\noindent {\bf Planar separator.}  $\vor(S)$ is a planar graph of $O(mn)$ size with all Voronoi edge bends and Voronoi vertices as graph vertices.  By a recursive application of the planar separator theorem, one can produce an \emph{$m^2$-division} of $\vor(S)$: it is divided into $O(n/m)$ regions, each region contains $O(m^2)$ vertices, and the boundary of each region contains $O(m)$ vertices~\cite{F87}.  

Extract the subset $B \subset S$ of points whose Voronoi cell boundaries contain some region boundary vertices in the $m^2$-division.  So $|B| = O(m \cdot n/m) = O(n)$.   Compute $\vor(B)$ and triangulate it as in triangulating $\vor(S)$.  By our choice of $B$, the region boundaries in the $m^2$-division of $\vor(S)$ form a subgraph of $\vor(B)$.  Label in $O(n)$ time the Voronoi edge bends and Voronoi vertices in $\vor(B)$ whether they exist in $\vor(S)$.  

We construct point location data structures for every region $\Pi$ in the $m^2$-division as follows.  For every boundary vertex $w$ of $\Pi$, let $Q^*_w$ be the largest homothetic copy of $Q^*$ centered at $w$ such that $\Int(Q^*_w) \cap B = \emptyset$.  
These $Q^*_w$'s form an arrangement of $O(m^2)$ complexity, and we construct a point location data structure that allows a point to be located in this arrangement in $O(\log m)$ time.  We also construct a point location data structure for the portion of the triangulated $\vor(S)$ inside $\Pi$.  Since the region has $O(m^2)$ complexity, this point location data structure can return in $O(\log m)$ time the triangle in the triangulated $\vor(S)$ that contains a point inside $\Pi$.  

\vspace{6pt}

\noindent {\bf Output and performance.}  The following result summarizes the output and performance of the training phase.  Its proof is given in Appendix~\ref{app:train-0}.  The proof of Lemma~\ref{lem:train}(a) is similar to an analogous result for sorting in~\cite{cheng20b}.

\begin{lemma}
	\label{lem:train}
	Let $\mathscr{D}_a$, $a \in [m]$, be the distributions in the hidden mixture.  The training phase computes the  following structures in $O(mn\log^{O(1)}(mn) + m^\eps n^{1+\eps}\log(mn))$ time.
	\begin{enumerate}[\label=(\alph{enumi})]
		\item A set $S$ of $O(mn)$ points and $\vor(S)$.  It holds with probability at least $1-1/n$ that for any $a \in [1,m]$ and any $v \in V_S$, $\sum_{i=1}^n \mathrm{Pr}[X_{iv} \, | \, I \sim \mathscr{D}_a] = O(1/m)$, where $X_{iv} = 1$ if $p_i \in I$ conflicts with $v$ and $X_{iv} = 0$ otherwise.

		\item Point location structures $L_S$ and $L_i$ for each $i \in [n]$ that allow us to locate $p_i$ in the triangulated $\vor(S)$ in $O\bigl(\frac{1}{\varepsilon}H(t_i)\bigr)$ expected time, where $t_i$ is the random variable that represents the point location outcome, and $H(t_i)$ is the entropy of the distribution of $t_i$.

		\item A net-tree $T_S$ for $S$, the $O(n)$ clusters in $T_S$, the initially empty van Emde Boas trees  for the clusters, and the nearest common ancestor data structures for the clusters.

		\item An $m^2$-division of $\vor(S)$, the subset $B \subseteq S$ of $O(n)$ points whose Voronoi cell boundaries contain some region boundary vertices in the $m^2$-division, $\vor(B)$, and the point location data structures for the regions in the $m^2$-division.
		
	\end{enumerate}
\end{lemma}

Lemma~\ref{lem:train}(a) leads to Lemma~\ref{lem:R} below, which implies that for any $v \in V_S$, if we feed the input points that conflict with $v$ to a procedure that runs in quadratic time in the worst case, the expected running time of this procedure over all points in $V_S$ is $O(n)$.  The proof  of Lemma~\ref{lem:R} is just an algebraic manipulation of the probabilities, and it is given in Appendix~\ref{app:train}.

\begin{lemma}
	\label{lem:R}
	For every $v \in V_S$, let $Z_v$ be the subset of input points that conflict with $v$.  It holds with probability at least $1- O(1/n)$ that $\sum_{v \in V_S} \mathrm{E}\bigl[|Z_{v}|^2\bigr] = O(n)$.  
\end{lemma}

We state two technical results.  Their proofs are in Appendix~\ref{app:train}.  Figure~\ref{fg:tech}(a) and~(b) illustrate these two lemmas.

\begin{lemma}
	\label{lem:enclose}
	Consider $\vor(Y)$ for some point set $Y$.  For any point $x \in \mathbb{R}^2$, let $Q^*_x$ be the largest homothetic copy of $Q^*$ centered at $x$ such that $\Int(Q^*_x) \cap Y = \emptyset$.  Let $w_1$ and $w_2$ be two adjacent Voronoi edge bends or Voronoi vertices in $\vor(Y)$.   For any point $x \in w_1w_2$, $Q^*_x \subseteq Q^*_{w_1} \cup Q^*_{w_2}$.  The same property holds if $w_1$ and $w_2$ are Voronoi vertices connected by a Voronoi edge, and $x$ lies on that Voronoi edge.
\end{lemma}

\begin{figure}
	\centering
	\begin{tabular}{ccc}
		\includegraphics[scale=0.45]{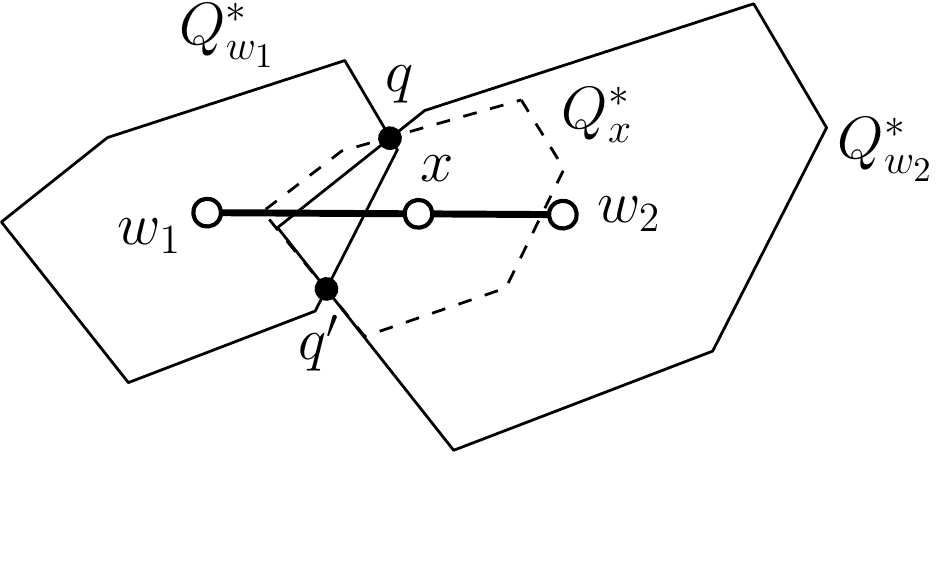} & & 
		\includegraphics[scale=0.5]{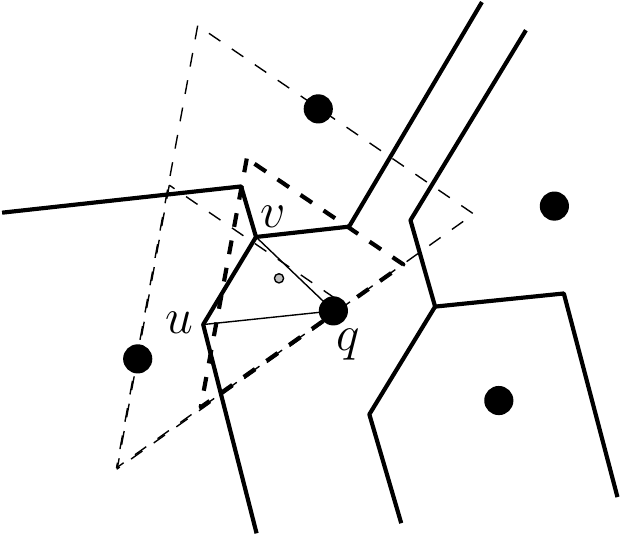} \\
		(a) & \hspace{.25in} & (b) \\
	\end{tabular}
	\caption{(a)~The points $q$ and $q'$ define a Voronoi edge, and $w_1$ and $w_2$ are two adjacent Voronoi edge bends or Voronoi vertices on this edge.  At any point $x$ between $w_1$ and $w_2$, the polygon $Q^*_x$ (shown dashed) is a subset of $Q^*_{w_1} \cup Q^*_{w_2}$. (b)~A triangle $quv$ in the triangulated Voronoi diagram in Figure~\ref{fg:example} is shown.  If a point $p$ conflicts with the white dot (i.e., lies inside the bold dashed triangle), then $p$ conflicts with $u$ or $v$ (i.e., lies inside one of the two light dashed circles.)}
	\label{fg:tech}
\end{figure}

\begin{lemma}
	\label{lem:conflict-tech}
	Let $q$ be a point in some point set $Y$.  Let $quv$ be a triangle in the triangulated $\vor(Y)$.  If a point $p \not\in Y$ conflicts with a point in $quv$, then $p$ conflicts with $u$ or $v$.  Hence, if $p$ conflicts with $V_q(Y)$, $p$ conflicts with a Voronoi edge bend or Voronoi vertex in $\partial V_q(Y)$.
\end{lemma}
\cancel{
\begin{proof}
	For any point $y \in \mathbb{R}^2$, let $Q^*_y$ the largest homothetic copy of $Q^*$ centered at $y$ such that $\Int(Q^*_y) \cap Y = \emptyset$.  It suffices to show that $p \in Q^*_u$ or $p \in Q^*_v$.  If $p$ conflicts with any point in $uv$, Lemma~\ref{lem:enclose} implies that $p \in Q^*_u$ or $p \in Q^*_v$.  Suppose that $p$ does not conflict with any point in $uv$.  So $V_p(Y \cup \{p\})$ does not intersect $uv$.  Since $V_q(Y \cup \{p\})$ is star-shaped with respect to $q$, $V_p(Y \cup \{p\})$ cannot lie strictly inside $quv$ because $quv \setminus V_p(Y \cup \{p\})$ is not star-shaped with respect to $q$ otherwise.  So $V_p(Y \cup \{p\})$ must cross $qu$ or $qv$.  Without loss of generality, assume that $V_p(Y \cup \{p\})$ intersects $qu$ at some point $x$.   As $x \in V_q(Y)$, $q \in \partial Q^*_x$.  As $x \in V_p(Y \cup \{p\})$,  $d_{Q^*}(x,p) = d_Q(p,x) \leq d_q(q,x) = d_{Q^*}(x,q)$, which implies that $p \in Q^*_x$.  Since $q$, $x$ and $u$ are collinear, $Q^*_x \subseteq Q^*_{u}$.   
\end{proof}
}

\section{Operation phase}
	 
	Given an instance $I = (p_1,\cdots,p_n)$, we construct $\vor(I)$ using the pseudocode below.
	 \begin{quote}
	 \noindent {\sc Operation Phase}
	 \begin{enumerate}
	 	
	 	\item For each $i \in [n]$, query $L_i$ to find the triangle $t_i$ in  the triangulated $\vor(S)$ that contains $p_i$, and if the search fails, query $L_S$ to find $t_i$.
	 	
	 	\item For each $i \in [n]$, search $\vor(S)$ from $t_i$ to find $V_S|_{p_i}$, i.e., the subset of $V_S$ that conflict with $p_i$.  This also gives the subset of $S$ whose Voronoi cells conflict with the input points.  Let $R$ be the union of this subset of $S$ and the set of representative points of all cluster roots in $T_S$.
	
		\item Compute the compression $T_R$ of $T_S$ to $R$.
		
		\item Construct the nearest neighbor graph 1-$\NN_R$ under the metric $d$ from $T_R$.
		
		\item Compute $\vor(R)$ from 1-$\NN_R$.
		
	 	\item Modify $\vor(R)$ to produce $\vor(R \cup I)$.
	 	
	 	\item Split $\vor(R \cup I)$ to produce $\vor(I)$ and $\vor(R)$.  Return $\vor(I)$.
	 	
	 \end{enumerate}
 	\end{quote}


We analyze step~1 in Section~\ref{sec:locate}, steps~2~and~3 in Section~\ref{sec:construct}, steps~4~and~5 in Section~\ref{sec:relax}, and steps~6~and~7 in Section~\ref{sec:I}.  Step~1 is the most time-consuming; all other steps run in $O(n)$ expected time or $O(n\log\log m)$ expected time.
 

\subsection{Point location}
\label{sec:locate}

By Lemma~\ref{lem:train}(b), step~1 runs in $O\bigl(\sum_{i=1}^n \frac{1}{\eps}H(t_i)\bigr)$ expected time, which is $O\bigl(\frac{1}{\eps}n\log m + \frac{1}{\eps}H(t_1,\ldots,t_n)\bigr)$ as we will show later.  By Lemma~\ref{lem:convert} below, if there is an algorithm that can use $\vor(I)$ to determine $t_1,\ldots,t_n$ in $c(n)$ expected time, then $H(t_1,\ldots,t_n) = O(c(n) + H)$, implying that step~1 takes $O\big(\frac{1}{\eps}(n\log m + c(n) + H)\bigr)$ expected time.  Any preprocessing cost of $S$ is excluded from $c(n)$.  We present such an algorithm.   

\begin{lemma}[Lemma~2.3~in~\cite{ailon11}]
	\label{lem:convert}
	Let $\mathscr{D}$ be a distribution on a universe $\cal U$.  Let $X : {\cal U} \rightarrow {\cal X}$, and let $Y : {\cal U} \rightarrow {\cal Y}$ be two random variables.  Suppose that there is a comparison-based algorithm that computes a function $f : (I,X(I)) \rightarrow Y(I)$ in $C$ expected comparisons over $\mathscr{D}$ for every $I \in {\cal U}$.  Then $H(Y) = C + O(H(X))$.
\end{lemma}

Recall that we have computed in the training phase the subset $B \subseteq S$ whose Voronoi cell boundaries contain some region boundary vertices in the $m^2$-division of $\vor(S)$.   Note that $|B| = O(n)$.   We have also computed $\vor(B)$ and point location data structures associated with the regions in the $m^2$-division.  We use $\vor(B)$ and these point location data structures determines $t_1, \ldots, t_n$ as follows.  

\begin{itemize}
	\item Task~1: Merge $\vor(B)$ with $\vor(I)$ to form the triangulated $\vor(B \cup I)$.   
	\item Task~2: \parbox[t]{4.8in}{Use $\vor(S)$, $\vor(B)$, and $\vor(B\cup I)$ to find the triangles $t_1,\ldots,t_n$.}
\end{itemize}

\noindent We discuss these two tasks in the following.

\cancel{
\begin{figure}
	\centering
	\begin{tabular}{ccc}
		\includegraphics[scale=0.35]{sep-1} & & 
		\includegraphics[scale=0.35]{sep-2} \\
		(a) & \hspace*{.2in} & (b) \\
	\end{tabular}
\caption{(a)~A schematic drawing of a Euclidean Voronoi diagram.  The white dots are the region boundary vertices, and the edges between them form the region boundaries.  (b)~If a point $p$ conflicts with a point in $\Pi$, but $p$ does not in any of the largest empty circles that are centered at the boundary vertices of $\Pi$, then $p$ lies in $\Pi$.}
\label{fg:sep}
\end{figure}
}


\vspace{6pt}

\noindent {\bf Task~1.}
For every point $p \in B$, define a polygonal cone surface $C_p = \bigl\{(a,b,d_Q(p,(a,b)) : (a,b) \in \mathbb{R}^2\bigr\}$.  Each horizontal cross-section of $C_p$ is a scaled copy of $Q$ centered at $p$.  The triangulated $\vor(B)$ is the vertical projection of the lower envelope of $\{C_p : p \in B\}$, denoted by ${\cal L}(B)$.  Similarly, ${\cal L}(I)$ projects to $\vor(I)$.  We take the lower envelope of ${\cal L}(B)$ and ${\cal L}(I)$ to form ${\cal L}(B \cup I)$ which projects to $\vor(B \cup I)$.  We do so in $O(n2^{O(\log^* n)})$ expected time with a randomized algorithm that is based on an approach proposed and analyzed by Chan~\cite[Section~4]{chan16}.   More details are given in Appendix~\ref{app:step1-1}.

\vspace{6pt}

\noindent {\bf Task~2.}  
Suppose that for an input point $p_i \in I$, we have determined some subset $B_i$ that satisfies $B \subseteq B_i \subseteq S$, and we have computed a Voronoi edge bend or Voronoi vertex $v_i$ in $\vor(B_i)$ that conflicts with $p_i$ and is known to be in $V_S$ or not.

If $v_i \in V_S$, we search $\vor(S)$ from $v_i$ to find $V_S|_{p_i}$ (i.e., the subset of $V_S$ that conflict with $p_i$), which by Lemma~\ref{lem:enclose} also gives the triangle $t_i$ in the triangulated $\vor(S)$ that contains $p_i$.  By Lemma~\ref{lem:R}, the expected total running time of this procedure over all input points is $O(n)$.  

Suppose that $v_i \not\in V_S$.  So $v_i$ is not a region boundary vertex in the $m^2$-division of $\vor(S)$, i.e., $v_i$ lies inside a region in the $m^2$-division of $\vor(S)$, say $\Pi$.   For each boundary vertex $w$ of $\Pi$, let $Q^*_w$ be the largest homothetic copy of $Q^*$ centered at $w$ such that $\Int(Q^*_w) \cap B = \emptyset$.  
These $Q^*_w$'s form an arrangement of $O(m^2)$ complexity, and we locate $p_i$ in this arrangement in $O(\log m)$ time.  It tells us whether $p_i \in Q^*_w$ for some boundary vertex $w$ of $\Pi$.  If so, then $p_i$ conflicts with $w$, which belongs to $V_S$, and we search $\vor(S)$ from $w$ to find $V_S|_{p_i}$ and hence the triangle $t_i$ in the triangulated $\vor(S)$ that contains $p_i$.  Otherwise, $p_i$ must lie inside $\Pi$ in order to conflict with $v_i$ inside $\Pi$ without conflicting with any boundary vertex of $\Pi$.  So we do a point location in $O(\log m)$ time to locate $p_i$ in the portion of the triangulated $\vor(S)$ inside $\Pi$.  This gives $t_i$.  


How do we compute $v_i$ for $p_i$?   We discuss this computation and provide more details of Step~2 in Appendix~\ref{app:step1-2}.  The following lemma summarizes the result that follows from the discussion above.

\begin{lemma}
	\label{lem:alg}
	Given $\vor(I)$, the triangles $t_1,\ldots, t_n$ in the triangulated $\vor(S)$ that contain $p_1,\ldots, p_n \in I$ can be computed in $O\left(n\log m + n2^{O(\log^* n)}\right)$ expected time.
\end{lemma}

\begin{lemma}
	\label{lem:locate}
	Step~$1$ of the operation phase takes $O\bigl(\frac{1}{\eps}(n\log m \! + \! n2^{O(\log^* n)} \! + \!\! H)\bigr)$ expected time, where $H$ is the entropy of the distribution of $\vor(I)$.
\end{lemma}
\begin{proof}
	Let $A \in [1,m]$ be a random variable that indicates which distribution in the mixture generates the input instance.  By the chain rule for conditional entropy~\cite[Proposition~2.23]{ray}, $H(t_i) \leq H(t_i) + H(A|t_i) = H(t_i,A) = H(A) + H(t_i|A)$.  It is known that $H(A) \leq \log_2 (\text{domain size of $A$}) = \log_2 m$~\cite[Theorem~2.43]{ray}.  Thus, $\sum_{i=1}^n H(t_i) \leq n\log_2 m + \sum_{i=1}^n H(t_i |A)$.  The variables $t_1|A, \ldots, t_n|A$ are mutually independent.  So $\sum_{i=1}^n H(t_i|A) = H(t_1,\ldots,t_n|A)$.  Since entropy is not increased by conditioning~\cite[Theorem~2.38]{ray}, we get $\sum_{i=1}^n H(t_i|A) = H(t_1,\ldots,t_n|A) \leq H(t_1,\ldots,t_n)$.   By Lemma~\ref{lem:alg}, we can determine $t_1,\ldots, t_n$ using $\vor(I)$ in $O(n\log m + n2^{O(\log^* n)})$ expected time.  So $H(t_1,\ldots,t_n) = O(n\log m + n2^{O(\log^* n)} + H)$ by Lemma~\ref{lem:convert}, where $H$ is the entropy of the distribution of $\vor(I)$. 
\end{proof}

In the Euclidean metric, merging $\mathrm{Vor}(B)$ and $\mathrm{Vor}(I)$ into $\mathrm{Vor}(B \cup I)$ can be reduced to finding the intersection of two convex polyhedra of $O(n)$ size in $\mathbb{R}^3$, which can be solved in $O(n)$ time~\cite{chan16}.  So the expected running time of step~1 improves to $O\bigl(\frac{1}{\eps}(n\log m + H)\bigr)$.

\subsection{Construction of $\pmb R$}
\label{sec:construct}

Step~1 determines the triangle $t_i$ in the triangulated $\vor(S)$ that contains $p_i \in I$.  We search $\vor(S)$ from $t_i$ to find $V_S|_{p_i}$, which takes $O\bigl(\bigl|V_S|_{p_i}\bigr|\bigr)$ time~\cite{KMM93}.  This search also gives the Voronoi cells that conflict with $p_i$.   The total time over all $i \in [n]$ is $O\bigl(\sum_{v \in V_S} |Z_v|\bigr)$, where $Z_v$ is the subset of input points that conflict with $v$.   Since $R$ includes all sites whose cells conflict with the input points and the representative points of all cluster roots in $T_S$, we have $|R| \leq \sum_{v \in V_S} |Z_v| + O(n)$.  The following result follows from Lemma~\ref{lem:R}.

\begin{lemma}
	The set $R$ has $O(n)$ expected size.  Step~2 of the operation phase constructs $R$ in $O(n)$ expected time.
\end{lemma}

\subsection{Extraction of $\pmb{\vor(R)}$}
\label{sec:relax}

\subsubsection{Construction of $\pmb{T_R}$}
\label{sec:relax-extract}

We define a \emph{compression} of a net-tree $T$.   Select a subset $U$ of leaves in $T$.  Let $T' \subseteq T$ be the minimal subtree that spans $U$.  Bypass all internal nodes in $T'$ that have only one child. \emph{The resulting tree is the compression of $T$ to $U$.}   The following result is an easy observation.

\begin{lemma}
	\label{lem:easy}
	Let $T$ be a net-tree.  Let $T_1$ be the compression of $T$ to a subset $U_1$ of leaves.  The compression of $T_1$ to any subset $U_2$ of leaves in $T_1$ can also be obtained by a compression of $T$ to $U_2$.
\end{lemma}

Conceptually, $T_R$ is defined as follows.  Select all leaves of $T_S$ that are points in $R$, and $T_R$ is the compression of $T_S$ to these selected leaves.  Since $R$ includes the representative points of all cluster roots, all ancestors of the cluster roots in $T_S$ will survive the compression and exist as nodes in $T_R$.  The compression affects the clusters only.  More precisely, for each cluster $c$ in $T_S$, we select its leaves that are points in $R$ and compute the compression $T_c$ of the cluster $c$ to these selected leaves.  Substituting every cluster $c$ in $T_S$ by $T_c$ gives the desired $T_R$.  It remains to discuss how to compute the $T_c$'s.

We divide $R$ in $O(n)$ expected time into sublists $R_1, R_2 ,\ldots$ such that $R_c$ consists of the points that are leaves in cluster $c$.  Recall that  every cluster $c$ has an initially empty van Emde Boas tree $\mathit{EB}_c$ for its leaves in left-to-right order.  For each $R_c$, we insert all leaves in $R_c$ into $\mathit{EB}_c$ and then repeatedly perform extract-min on $\mathit{EB}_c$.  This gives in $O(|R_c|\log\log m)$ time a sorted list $R'_c$ of the leaves in $R_c$ according to their left-to-right order in the cluster $c$.  

If $|R'_c| = 1$, then $T_c$ consists of the single leaf in $R_c$.  Suppose that $|R'_c| \geq 2$.  We construct $T_c$ using a stack.  Initially, $T_c$ is a single node which is the first leaf in  $R'_c$.  The stack stores the nodes on the rightmost root-to-leaf path in the current $T_c$, with the root at the stack bottom and the leaf at the stack top.  When we scan the next leaf $q$ in $R'_c$, we find in cluster $c$ the nearest common ancestor $x$ of $q$ and $q$'s predecessor in $R'_c$.  This takes $O(\log\log m)$ time~\cite{TL88}.   If we see $x$ at the stack top, we add $q$ as a new leaf to $T_c$ with $x$ as its parent, and then we push $q$ onto the stack.  Refer to the left image in Figure~\ref{fg:stack}.   If we see an ancestor $z$ of $x$ at the stack top, let $y$ be the node that was immediately above $z$ in the stack and was just popped, we make $x$ the rightmost child of $z$ in $T_c$ (which was $y$ previously), we also make $y$ and $q$ the left and right children of $x$ respectively, and then we push $x$ and $q$ in this order onto the stack.  Refer to the middle image in Figure~\ref{fg:stack}.  If neither of the two conditions above happens and the stack is not empty, we pop the stack and repeat.  Refer to the right image in Figure~\ref{fg:stack}.  If the stack becomes empty, we make $x$ the new root of $T_c$,  we also make the old root of $T_c$ and $q$ the left and right children of $x$ respectively, and then we push $x$ and $q$ in this order onto the stack.  The construction of $T_c$ takes $O(|R_c|\log\log m)$ time.  

\begin{figure}
	\centerline{\includegraphics[scale=0.6]{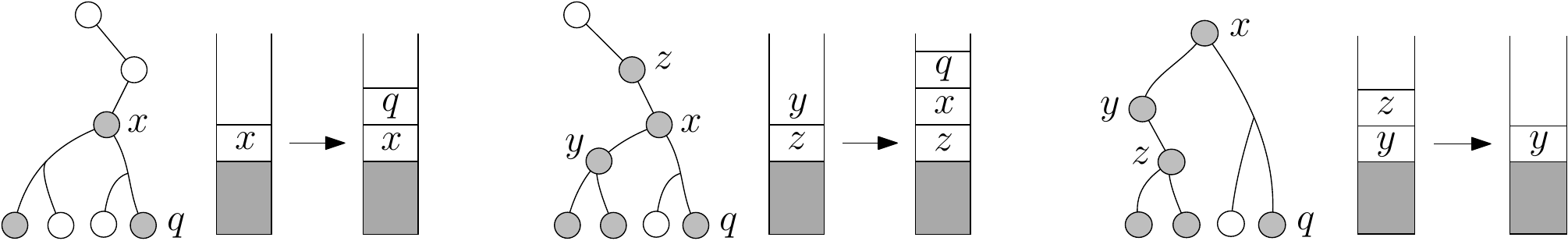}}
	\caption{Three different cases in the manipulation of the stack.  The tree shown is a part of $T_S$.  The gray nodes are nodes in $T_c$.  The gray leaves are leaves in $R_c$.}
	\label{fg:stack}
\end{figure}


\begin{lemma}
	\label{lem:relaxed}
	The compression $T_R$ of $T_S$ to $R$ can be computed in $O(n\log\log m)$ time.
\end{lemma}

\subsubsection{Construction of the $\pmb k$-nearest neighbor graph}
\label{sec:kNN}

Let $X$ be any subset of $S$.  Assume that the compression $T_X$ of $T_S$ to $X$ is available.  We show how to use $T_X$ to construct in $O(k|X|)$ time the $k$-nearest neighbor graph of $X$ under the metric $d$.  We denote this graph by $k$-$\NN_X$.  We will use the \emph{well-separated pair decomposition} or WSPD for short.  For any $c \geq 1$, a set $\bigl\{\{A_1,B_1\},\ldots, \{A_s,B_s\}\bigr\}$ is a $c$-WSPD of $X$ under $d$ if the following properties are satisfied:
\begin{itemize}
	\item $\forall \, i, \,\,\, A_i, B_i \subseteq X$.
	\item $\forall\, \text{distinct} \, x,y \in X$, $\exists \, i$ such that $\bigl\{x,y\} \in \bigl\{\{a,b\} : a \in A_i \wedge b \in B_i\bigr\}$.
	\item $\forall \, i$, the maximum of the diameters of $A_i$ and $B_i$ under $d$ is less than $\frac{1}{c} \cdot d(A_i,B_i)$.  It implies that $A_i \cap B_i = \emptyset$.
\end{itemize}
It is known that a $c$-WSPD has $O(c^{O(1)}|X|)$ size and can be constructed in $O(c^{(O(1)}|X|)$ time from a net-tree for $X$~\cite{har-peled06}.  The same method works for a compression $T_X$ of $T_S$ to $X$, giving a $c$-WPSD of $O((c+1)^{O(1)}|X|)$ size in $O((c+1)^{O(1)}|X|)$ time.  The details of the WSPD construction are given in Appendix~\ref{app:NN-1}.
To compute $k$-$\NN_X$, we transfer a strategy in~\cite{CK95} for constructing a Euclidean $k$-nearest neighbor graph using a WSPD.
The details are given in Appendix~\ref{app:NN-2}.

\begin{lemma}
	\label{lem:knn}
	Given the compression $T_X$ of $T_S$ to any subset $X \subseteq S$, the $k$-$\NN_X$ can be constructed in $O(k|X|)$ time.
\end{lemma}

The next result shows that the vertex degree of $1$-$\NN_X$ is $O(1)$.  Its proof is given in Appendix~\ref{app:NN-3} which is adapted from an analogous result in the Euclidean case~\cite{miller97}.

\begin{lemma}
	\label{lem:deg}
	For any subset $X \subseteq S$, every vertex in $1$-$\NN_X$ has $O(1)$ degree, and adjacent vertices in $1$-$\NN_X$ are Voronoi neighbors in $\vor(X)$.
\end{lemma}

\subsubsection{$\pmb{\vor(R)}$ from the nearest neighbor graph}
\label{sec:vorNN}

We show how to construct $\vor(R)$ in $O(n)$ expected time using 1-$\NN_R$.  We use the following recursive routine which is similar to the one in~\cite{BM11} for constructing an Euclidean Delaunay triangulation from the Euclidean nearest neighbor graph.  The top-level call is VorNN$(R,T_R)$.  

\begin{quote}
\noindent VorNN$(Y,T_Y)$
\begin{enumerate}
	\item If $|Y| = O(1)$, compute $\vor(Y)$ directly and return.
	\item Compute 1-$\NN_Y$ under the metric $d$ using $T_Y$.
	\item Let $X \subseteq Y$ be a random sample such that $X$ meets every connected component of 1-$\NN_Y$, and $\mathrm{Pr}[p \in X] = 1/2$ for every $p \in Y$.
	\item Compute the compression $T_X$ of $T_Y$ to $X$.
	\item Call VorNN$(X,T_X)$ to compute $\vor(X)$.
	\item Using 1-$\NN_Y$ as a guide, insert the points in $Y \setminus X$ into $\vor(X)$ to form $\vor(Y)$.
\end{enumerate}
\end{quote}

There are two differences from~\cite{BM11}.  First, we use a compression $T_Y$ of $T_S$ to compute 1-$\NN_Y$ in step~2, which takes $O(|Y|)$ time by Lemma~\ref{lem:knn}.  Second, we need to compress $T_Y$ to $T_X$ in step~4.  This compression works in almost the same way as described in Section~\ref{sec:relax-extract} except that we can afford to traverse $T_Y$ in $O(|Y|)$ time to answer all nearest common ancestor queries required for constructing $T_X$.  Thus, step~4 runs in $O(|Y|)$ time.

Step~3 is implemented as follows~\cite{BM11}.  Form an arbitrary maximal matching of 1-$\NN_Y$.  By the definition of 1-$\NN_Y$, each connected component of 1-$\NN_Y$ contains at least one matched pair.  Randomly select one point from every matched pair.    Then, among those unmatched points in 1-$\NN_Y$, select each one with probability  1/2 uniformly at random.  The selected points form the subset $X$ required in step~3.  The time needed is $O(|Y|)$.  

In step~6, for each $p \in Y \setminus X$ that is connected to some point $q \in X$ in 1-$\NN_Y$, $p$ and $q$ are Voronoi neighbors in $\vor(Y)$ by Lemma~\ref{lem:deg}.  So $p$ conflicts with a point in $V_q(X)$.
By Lemma~\ref{lem:conflict-tech}, $p$ conflicts with a Voronoi edge bend or Voronoi vertex in $\partial V_q(X)$, which can be found in $O\bigl(\bigl|\partial V_q(X)\bigr|\bigr)$ time.  
After finding a Voronoi edge bend or Voronoi vertex $v$ in $\partial V_q(X)$ that conflicts with $p$, we search $\vor(X)$ from $v$ to find all Voronoi edge bends and Voronoi vertices that conflict with $p$.  In the same search of $\vor(X)$ , we modify $\vor(X)$ into $\vor\bigl(X \cup \{p\}\bigr)$ as in a randomized incremental construction~\cite{KMM93}.  By the Clarkson-Shor analysis~\cite{CS89}, the expected running time of the search of $\vor(X)$ and the Voronoi diagram modification over the insertions of all points in $Y \setminus X$ is $O(|Y|)$.  We spend $O\bigl(\bigl|\partial V_q(X)\bigr|\bigr)$ time to find $v$.  It translates to an $O(1)$ charge at each vertex of $V_q(X)$.  This charging happens only for $q$'s neighbors in 1-$\NN_Y$.  By Lemma~\ref{lem:deg}, there are $O(1)$ such neighbors of $q$, so the charge at each vertex of $V_q(X)$ is $O(1)$.  Moreover, if a vertex of $V_q(X)$ is destroyed by the insertion of a point from $Y \setminus X$, that vertex will not reappear.  So the $O\bigl(\bigl|\partial V_q(X)\bigr|\bigr)$ cost is absorbed by the structural changes which is already taken care of by the Clarkson-Shor analysis.  Unwinding the recursion gives a total expected running time of $O(|R| + |R|/2 + |R|/4 + \cdots) = O(|R|)$.  For completeness, more details of the analysis given in~\cite{BM11} can be found in Appendix~\ref{app:vorNN}.

\begin{lemma}
	\label{lem:delNN}
	\emph{VorNN}$(R,T_R)$ computes $\vor(R)$ in $O(|R|)$ expected time.
\end{lemma}

\subsection{Computing $\pmb{\vor(I)}$ from $\pmb{\vor(R)}$ and $\pmb I$}
\label{sec:I}


Let $q$ be a point in $R$.  Let $v_1, v_2, \ldots$ be the vertices of $V_q(R)$, in clockwise order, which may be Voronoi edge bends or Voronoi vertices.   Let $Q^*_{v_i}$ denote the largest homothetic copy of $Q^*$ centered at $v_i$ such that $\Int(Q^*_{v_i}) \cap R = \emptyset$.  Let $Z_{v_i} = Q^*_{v_i} \cap I$ where $I$ is an input instance.

\begin{lemma}
	\label{lem:limit}
   The portions of $\vor(R \cup I)$ and $\vor\bigl(\{q\} \cup Z_{v_i}\cup Z_{v_{i+1}}\bigr)$ inside the triangle $qv_iv_{i+1}$ are identical.
\end{lemma}
\begin{proof}
   Let $p$ be a point in $(R \cup I) \setminus \{q\}$ that contributes to $\vor(R \cup I)$ inside $qv_iv_{i+1}$.  As $qv_iv_{i+1} \subseteq V_q(R)$, $p \not\in R$.  So $p \in I$.  By Lemma~\ref{lem:conflict-tech}, $p$ conflicts with $v_i$ or $v_{i+1}$.
   \cancel{
   It suffices to show that $p \in Q^*_{v_i}$ or $p \in Q^*_{v_{i+!}}$.  
   If $p$ conflicts with any point in $v_iv_{i+1}$, Lemma~\ref{lem:enclose} implies that $p \in Q^*_{v_i}$ or $p \in Q^*_{v_{i+1}}$.  Suppose that $p$ does not conflict with any point in $v_iv_{i+1}$.  So $V_p(R \cup \{p\})$ does not intersect $v_iv_{i+1}$.  Since $p$ conflicts with $V_q(R)$, $p$ must conflict with a Voronoi edge bend or Voronoi vertex in $\partial V_q(R)$ by Lemma~\ref{lem:conflict}.   So $V_p(R \cup \{p\})$ must cross $qv_i$ or $qv_{i+1}$ in order to get out of $qv_iv_{i+1}$.  Without loss of generality, assume that $V_p(R \cup \{p\})$ intersects $qv_i$ at some point $x$.   Let $Q^*_x$ be the largest homothetic copy of $Q^*$ centered at $x$ such that $\Int(Q^*_x) \cap R = \emptyset$.   As $x \in V_q(R)$, $q \in \partial Q^*_x$.  As $x \in V_p(R \cup \{p\})$,  $d_{Q^*}(x,p) = d_Q(p,x) \leq d_q(q,x) = d_{Q^*}(x,q)$, which implies that $p \in Q^*_x$.  Since $q$, $x$ and $v_i$ are collinear, $Q^*_x \subseteq Q^*_{v_i}$.   
}
\end{proof}

Step~2 of the operation phase has found $V_S|_{p_i}$ for each $p_i \in I$.  $V_S|_{p_i}$ and the portions of the Voronoi edges of $\vor(S)$ among the points in $V_S|_{p_i}$ are preserved in $\vor(R)$ because $R$ includes the subset of $S$ whose  Voronoi cells conflict with the input points.  Hence, $\bigcup_{i=1}^n V_S|_{p_i}$ is  the set $U_R$ of Voronoi edge bends and Voronoi vertices in $\vor(R)$ that conflict with the input points (refer to Appendix~\ref{app:I} for a proof).  By Lemma~\ref{lem:limit}, we locally compute pieces of $\vor(R \cup I)$ and stitch them together.  The running time is  $O\bigl(\sum_{u,v} (|Z_u| + |Z_v|)\log (|Z_u| + |Z_v|)\bigr)$, where the sum is over all pairs $\{u,v\}$ of adjacent Voronoi edge bends and Voronoi vertices in $\vor(R)$ such that $\{u,v\} \cap U_R \not= \emptyset$.  Since the degrees of Voronoi edge bends and Voronoi vertices are two and three respectively, this running time can be bounded by $O\big(\sum_{v \in U_R} |Z_v|\log |Z_v|\bigr)$.  Since $U_R \subseteq V_S$, by Lemma~\ref{lem:R},  step~6 of the operation phase computes $\vor(R \cup I)$ in $O(n)$ expected time.

\cancel{
\begin{lemma}
	\label{lem:step5}
	Step~6 of the operation phase computes $\vor(R \cup I)$ in $O(n)$ expected time.
\end{lemma}
}

In step~7, the splitting of $\vor(R \cup I)$ into $\vor(R) $ and $\vor(I)$ can be performed in $O(n)$ expected time by using the algorithm in~\cite{chazelle02} for splitting a Euclidean Delaunay triangulation.  That algorithm is combinatorial in nature.  It relies on the Voronoi diagram being planar and of $O(n)$ size,  all points having $O(1)$ degrees in the nearest neighbor graph, and that one can delete a site from a Voronoi diagram in time proportional to its number of Voronoi neighbors.  The first two properties hold in our case, and it is known how to delete a site from an abstract Voronoi diagram so that the expected running time is proportional to its number of Voronoi neighbors~\cite{JP18}.

\begin{lemma}
	Step~6 of the operation phase computes $\vor(R \cup I)$ in $O(n)$ expected time, and step~7  splits $\vor(R \cup I)$ into $\vor(I)$ and $\vor(R)$ in $O(n)$ expected time.
\end{lemma}

In summary, since steps~2-7 of the operation phase take $O(n)$ expected time, the limiting complexity is dominated by the $O\bigl(\frac{1}{\eps}n\log m + \frac{1}{\eps}n2^{O(\log^* n)} + \frac{1}{\eps}H\big)$ expected running time of step~1.  In the Euclidean case, step~1 runs faster in $O\bigl(\frac{1}{\eps}n\log m + \frac{1}{\eps}H\big)$ time.  

\begin{theorem}
	\label{thm:main}
	Let $Q$ be a convex polygon with $O(1)$ complexity.  Let $n$ be the input size.  For any $\eps \in (0,1)$ and any hidden mixture of at most $m = o(\sqrt{n})$ product distributions such that each distribution contributes an instance with a probability of $\Omega(1/n)$, there is a self-improving algorithm for constructing a Voronoi diagram under $d_Q$ with a limiting complexity of $O(\frac{1}{\eps}n\log m + \frac{1}{\eps}n2^{O(\log^* n)} + \frac{1}{\eps}H)$.   For the Euclidean metric, the limiting complexity is $O(\frac{1}{\eps}n\log m + \frac{1}{\eps}H)$.  The training phase runs in $O(mn\log^2 (mn) + m^{\eps} n^{1+\eps}\log(mn))$ time.  The success probability is at least $1-O(1/n)$.
\end{theorem}

\section{Conclusion}

It is open whether one can get rid of the requirement that each distribution in the mixture contributes an instance with a probability of $\Omega(1/n)$, which is not needed for self-improving sorting~\cite{cheng20b}.  Eliminating the $n2^{O(\log^* n)}$ term from the limiting complexity might require solving the question raised in~\cite{chan16} that whether there is an $O(n)$-time algorithm for computing the lower envelope of pseudo-planes.  
As a Voronoi diagram can be interpreted as the lower envelope of some appropriate surfaces, a natural question is what surfaces admit a self-improving lower envelope algorithm.

\newpage

\bibliographystyle{plainurl}


\appendix

\cancel{
\section{Figures}

\begin{figure}[h]
	\centerline{\includegraphics[scale=0.7]{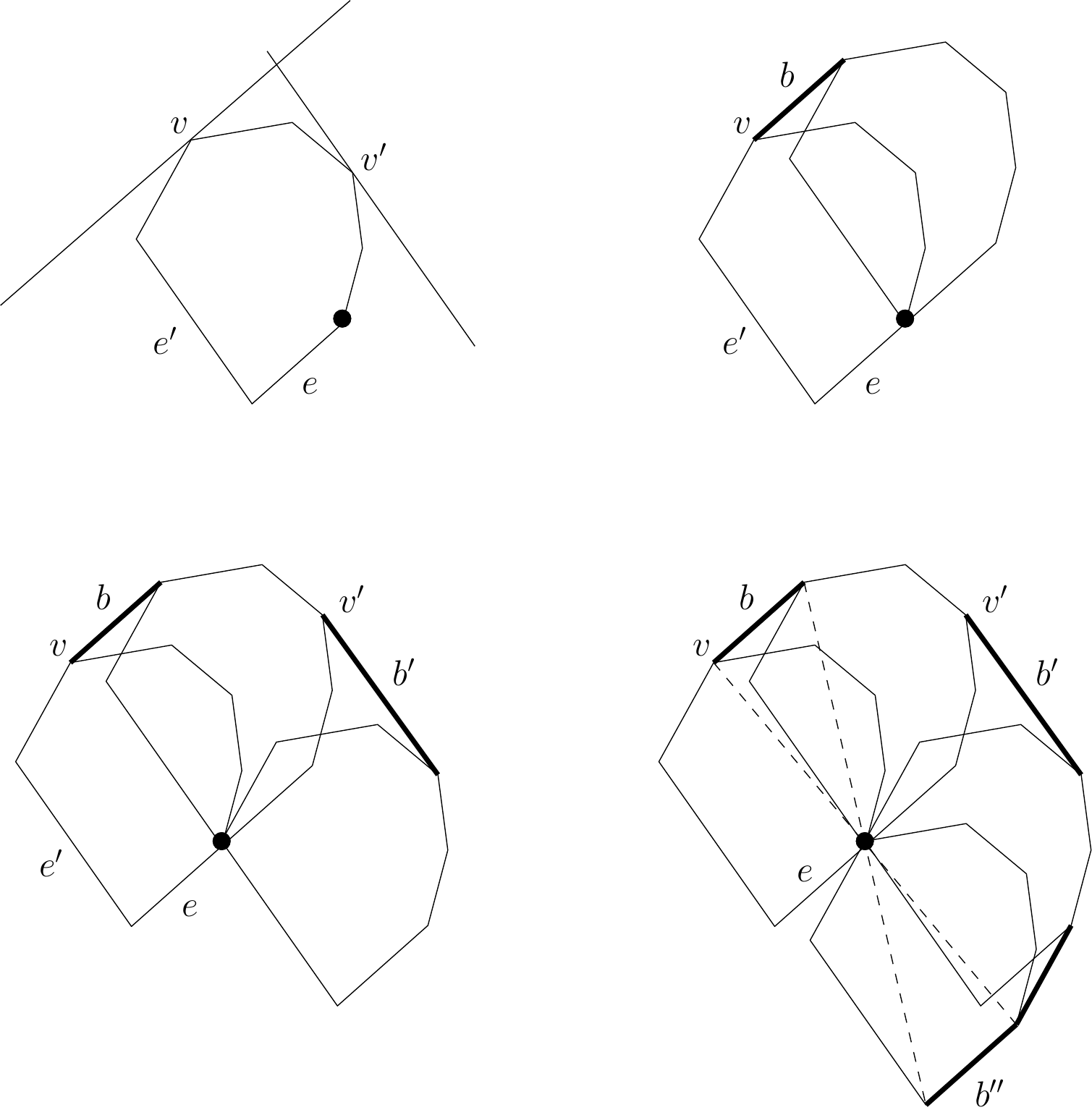}}
	\caption{The convex polygon shown is $Q^*$.  The vertices $v$ and $v'$ are the tangential contacts that the lines parallel to $e$ and $e'$ make with $Q^*$, respectively.  The black dot denotes the origin which is stationary.  In the upper right figure, when the origin slides from the upper endpoint of $e$ to its lower endpoint, the vertex $v$ sweeps out the edge $b$ in $\partial \hat{Q}$ which is a translate of $e$.  Similarly, in the lower left figure, $v'$ sweeps an edge $b'$ in $\partial \hat{Q}$ which is a translate of $e'$.  In the lower right figure, when the origin slides to $v$, a translate $b''$ of $e$ appears in $\partial \hat{Q}$.  The edges $b$ and $b''$ in $\partial \hat{Q}$ show that the origin is the center of symmetry.}
	\label{fg:slide}
\end{figure}
}

\section{Proof of Lemma~\ref{lem:train}}
\label{app:train-0}

We restate Lemma~\ref{lem:train} and then give its proof.

\vspace{8pt}

\noindent {\bfseries\sffamily Statement of Lemma~\ref{lem:train}:}\hspace{4pt}	\emph{Let $\mathscr{D}_a$, $a \in [m]$, be the distributions in the hidden mixture.  The training phase computes the  following structures in $O(mn\log^{O(1)}(mn) + m^\eps n^{1+\eps}\log(mn))$ time.}
\begin{enumerate}[\label=(\alph{enumi})]
	
	\item \emph{A set $S$ of $O(mn)$ points and $\vor(S)$.  It holds with probability at least $1-1/n$ that for any $a \in [1,m]$ and any $v \in V_S$, $\sum_{i=1}^n \mathrm{Pr}[X_{iv} \, | \, I \sim \mathscr{D}_a] = O(1/m)$, where $X_{iv} = 1$ if $p_i \in I$ conflicts with $v$ and $X_{iv} = 0$ otherwise.}
	
	\item \emph{Point location structures $L_S$ and $L_i$ for each $i \in [n]$ that allow us to locate $p_i$ in the triangulated $\vor(S)$ in $O\bigl(\frac{1}{\varepsilon}H(t_i)\bigr)$ expected time, where $t_i$ is the random variable that represents the point location outcome, and $H(t_i)$ is the entropy of the distribution of $t_i$.}
	
	\item \emph{A net-tree $T_S$ for $S$, the $O(n)$ clusters in $T_S$, the initially empty van Emde Boas trees  for the clusters, and the nearest common ancestor data structures for the clusters.}
	
	\item \emph{An $m^2$-division of $\vor(S)$, the subset $B \subseteq S$ of $O(n)$ points whose Voronoi cell boundaries contain some region boundary vertices in the $m^2$-division, $\vor(B)$, and the point location data structures for the regions in the $m^2$-division.}
	
\end{enumerate}
\begin{proof}
		Let $X = \{x_1,\ldots,x_{mn\ln(mn)}\}$ be the set of points from which the $\frac{1}{mn}$-net was extracted.  The set $S$ consists of this $\frac{1}{mn}$-net and the $O(1)$ dummy points.  Let $\sigma = \{j_1,j_2,j_3\} \subset [1,mn\ln(mn)]$ be a triple of distinct indices.
		Let $Q^*_{\sigma}$ be the homothetic copy of $Q^*$ that circumscribes $x_{j_1}$, $x_{j_2}$ and $x_{j_3}$ if it exists; otherwise, we ignore $\sigma$.  Assume that $\sigma$ is not ignored.  We analyze the number of points in any input instance that fall inside $Q^*_\sigma$.
		
		Fix any product distribution $\mathscr{D}_a$ in the hidden mixture.  Let ${\cal J}_{a,\sigma} = \{i \in [1,mn\ln(mn)]\setminus\sigma : \text{$x_i$ is drawn $\mathscr{D}_a$}\}$.   How large is ${\cal J}_{a,\sigma}$?  Since $\Pr[I \sim \mathscr{D}_a] = \Omega(1/n)$, the expected size of  ${\cal J}_{a,\sigma}$ is $\Omega(m\ln(mn))$.   Then, Chernoff bound implies that $|{\cal J}_{a,\sigma}| = \Omega(m\ln(mn))$ with probability at least $1 - (mn)^{-\Omega(m)}$.  
		
		For every $i \in {\cal J}_{a,\sigma}$, define $Y_{a,\sigma}(i )= 1$ if $x_i \in Q^*_{\sigma}$; otherwise, $Y_{a,\sigma}(i) = 0$.  Let $Y_{a,\sigma}= \sum_{i \in {\cal J}_{a,\sigma}} Y_{a,\sigma}(i)$.  
		The variables $Y_{a,\sigma}(i)$'s are independent from each other, so the Chernoff bound is applicable to $Y_{a,\sigma}$.  It says that for any $\lambda \in (0,1)$, 
		$\mathrm{Pr}\bigl[Y_{a,\sigma }> (1-\lambda)\mathrm{E}[Y_{a,\sigma}] \bigr] > 1 - e^{-\frac{1}{2}\lambda^2\mathrm{E}[Y_{a,\sigma}]}$.
		
		If $\mathrm{E}[Y_{a,\sigma}] > \frac{2}{\lambda^2(1-\lambda)}\ln(mn)$, then $\Pr\bigl[Y_{a,\sigma} > \frac{2}{\lambda^2}\ln(mn) \bigr] > 1 - (mn)^{-1/(1-\lambda)}$.  Setting $\lambda = 4/5$ gives $\mathrm{E}[Y_{a,\sigma}] > \frac{125}{8}\ln(mn) \Rightarrow \Pr \bigl[Y_{a,\sigma} > \frac{25}{8}\ln( mn) \bigr] > 1 - (mn)^{-5}$.  There are fewer than $m^3n^3\ln^3 (mn)$ triples of distinct indices.  By the union bound, it holds with probability at least $1 - \ln^3 (mn)/(m^2n^2) > 1 - 1/(mn)$ that for any triple $\sigma$ of distinct indices, if $\mathrm{E}[Y_{a,\sigma}] > \frac{125}{8}\ln (mn)$, then $Y_{a,\sigma} > \frac{25}{8}\ln (mn)$.
		
		Consider any Voronoi vertex $v \in V_S$ and its defining triple $\sigma$.  If $|Q^*_\sigma \cap X| \geq |X|/(mn) = 
		\ln(mn)$, then $Q^*_\sigma \cap S \not= \emptyset$ because $S$ is a $\frac{1}{mn}$-net of $X$.  But $Q^*_\sigma \cap S$ is empty as $v$ is a Voronoi vertex, which implies that $|Q^*_\sigma \cap X| < \ln(mn)$.  If we restrict our attention to instances in ${\cal J}_{a,\sigma}$ that  contribute to $Q^*_\sigma \cap X$, the count does not get bigger.  That is, $Y_{a,\sigma} \leq |Q^*_\sigma \cap X| < \ln(mn)$.  By the contrapositive of the result that we obtained earlier on the relation between $\mathrm{E}[Y_{a,\sigma}]$ and $Y_{a,\sigma}$, we conclude that $\mathrm{E}[Y_{a,\sigma}] \leq \frac{125}{8}\ln(mn)$.  Moreover, this upper bound on $\mathrm{E}[Y_{a,\sigma}]$ hold simultaneously for all defining triples of the Voronoi vertices in $V_S$ with probability at least $1 - 1/(mn)$.
		
		Since the input distribution is oblivious of the training and operation phases, we can use the instances in ${\cal J}_{a,\sigma}$ to derive the following inequality: $\EE{Y_{a,\sigma}} \geq  |{\cal J}_{a,\sigma}| \cdot \left(\sum_{i=1}^n \pr{X_{iv} \, | \, I \sim \mathscr{D}_a} \right) - 3$.
		The additive term of $-3$ stems from the fact that the indices in $\sigma$ are excluded from ${\cal J}_{a,\sigma}$ in the definition of $Y_{a,\sigma}$, but they are allowed in $|{\cal J}_{a,\sigma}|  \cdot \sum_{i=1}^n \pr{X_{iv} \, | \, I \sim \mathscr{D}_a}$.  Rearranging terms gives
		\[
		\sum_{i=1}^n \pr{X_{iv} \, | \, I \sim \mathscr{D}_a} = O\left(\frac{\EE{Y_{a,\sigma}}+ 3}{|{\cal J}_{a,\sigma}|} \right) = O\left(\frac{\EE{Y_{a,\sigma}} + 3}{m\ln(mn)}\right) = O\left(\frac{1}{m}\right).
		\]
		As discussed before, the above result holds for $\mathscr{D}_a$ with probability at least $1 - 1/(mn)$.  Applying the union bound over all $a \in [m]$, we get a success probability of at least $1 - 1/n$.
		
		Consider (b).  If $p_i$ falls into a triangle $t \in {\cal S}_i$ with weight $w_t$, the distribution-sensitive point location data structure~\cite{arya07,iacono04} ensures that the query time of $L_i$ is $O(\log (W/w_t))$, where $W = \sum_{t \in {\cal S}_i} w_t$.  Since $w_t$ is defined to be $\max\bigl\{(mn)^{-\eps}, \tilde{\pi}_{i,t}\bigr\}$ and the complexity of ${\cal S}_i$ is $O(m^\eps n^\eps)$, we have $W \leq \sum_{t \in {\cal S}_i} \bigl((mn)^{-\eps}+ \tilde{\pi}_{i,t}\bigr) = O(1)$.  Let $\pi_{i,t}$ be the true probability of $p_i$ hitting a triangle $t$ in the triangulated $\vor(S)$.  Using the Chernoff bound, one can prove as in~\cite[Lemma~3.4]{ailon11} that, with probability at least $1 - O(1/(mn))$, for every $i \in [n]$ and every $t$, if $\pi_{i,t} > (mn)^{-\eps/3}$, then $\tilde{\pi}_{i,t} \in [0.5\pi_{i,t},1.5\pi_{i,t}]$.  As $w_t = \max\bigl\{(mn)^{-\eps},\tilde{\pi}_{i,t}\bigr\}$, if $\pi_{i,t} > (mn)^{-\eps/3}$, the query time is $O(\log 1/w_t) = O(\log (1/\pi_{i,t}))$.  If $\pi_{i,t} \leq (mn)^{-\eps/3}$, we may query $L_S$ as well, so the query time is $O(\log (1/w_t)) + O(\log(mn))= O(\eps\log(mn)) + O(\log(mn)) =  O\bigl(\frac{1}{\eps}\log(1/\pi_{i,t})\bigr)$.  Therefore, the expected query time of $L_i$ is bounded by $O\left(\sum_{t \in {\cal S}_i} \pi_{i,t} \cdot \frac{1}{\eps}\log (1/\pi_{i,t}) \right) = \frac{1}{\eps}H(t)$.   
		
		The correctness of (c) and (d) follows from~\cite{F87,har-peled06} and our previous description.
	\end{proof}

\section{Missing details in Section~\ref{sec:train}}
\label{app:train}

We restate the results below and give their proofs.

\vspace{8pt}

\noindent {\bfseries\sffamily Statement of Lemma~\ref{lem:R}:}\hspace{4pt}\emph{For every $v \in V_S$, let $Z_v$ be the subset of input points that conflict with $v$.  It holds with probability at least $1- O(1/n)$ that $\sum_{v \in V_S} \mathrm{E}\bigl[|Z_{v}|^2\bigr] = O(n)$.}
\begin{proof}
	For every $i \in [n]$ and every $v \in V_S$, define $X_{iv} = 1$ if $p_i \in Z_{v}$ and $X_{iv} = 0$ otherwise.
	\begin{align*}
		&~~~~\sum_{v \in V_S}\E{|Z_v|^2} 
		= \EE{\sum_{v\in V_S} \left(\sum_{i \in [n]} X_{iv}\right)^2} = \sum_{v \in V_S} \sum_{i,j \in [n]} \E{X_{iv}X_{jv}} \\
		&= \sum_{a \in [m]} \sum_{v \in V_S} \sum_{i \in [n]} \pr{X_{iv} | I \sim {\cal D}_a} \cdot \pr{I \sim {\cal D}_a} + \\
		&\quad\quad \sum_{a \in [m]} \sum_{v \in V_S} \sum_{i \not= j} \pr{X_{iv}  \wedge X_{jv}| I \sim {\cal D}_a} \cdot \pr{I \sim {\cal D}_a} \\
		&= \sum_{a \in [m]} \pr{I \sim {\cal D}_a} \sum_{v \in V_S} O(1/m) + 
		\sum_{a \in [m]} \sum_{v \in V_S} \sum_{i \not= j} \pr{X_{iv}  \wedge X_{jv}| I \sim {\cal D}_a} \cdot \pr{I \sim {\cal D}_a} \\
		&= O(n) + \sum_{a \in [m]} \sum_{v \in V_S} \sum_{i \not= j} \pr{X_{iv}  \wedge X_{jv}| I \sim {\cal D}_a} \cdot \pr{I \sim {\cal D}_a}.
	\end{align*}
	Lemma~\ref{lem:train}(a) is invoked in the third step.  The last step is due to the fact that $|V_S| = O(mn)$ and $\sum_{a \in [m]} \pr{I \sim \mathscr{D}_a} = 1$.
	Under the condition that $I \sim \mathscr{D}_a$, $X_{iv}$ and $X_{jv}$ are independent.  Therefore, $\pr{X_{iv} \wedge X_{jv} | I \sim {\cal D}_a} = \pr{X_{iv} | I \sim {\cal D}_a} \cdot \pr{X_{jv} | I \sim {\cal D}_a}$.  As a result,
	\begin{align*}
		&~~~~\sum_{a \in [m]} \sum_{v \in V_S} \sum_{i \not= j} \pr{X_{iv} \wedge X_{jv} | I \sim {\cal D}_a} \cdot \pr{I \sim {\cal D}_a} \\
		&= \sum_{a \in [m]}  \pr{I \sim {\cal D}_a} \sum_{v \in V_S} \sum_{i \not= j} \pr{X_{iv} | I \sim {\cal D}_a} \cdot \pr{X_{jv} | I \sim {\cal D}_a}\\
		&\leq \sum_{a \in [m]}  \pr{I \sim {\cal D}_a} \sum_{v \in V_S} \Bigl(\sum_{i \in [n]} \pr{X_{iv} | I \sim {\cal D}_a} \Bigr)^2 \\
		&= \sum_{a \in [m]}  \pr{I \sim {\cal D}_a} \sum_{v \in V_S} O(1/m^2)  
		~=~O(n/m).
	\end{align*}
	In the last step, we use Lemma~\ref{lem:train}(a) and the relations that $|V_S| = O(mn)$ and $\sum_{a \in [m]} \pr{I \sim \mathscr{D}_a} = 1$.
\end{proof}


\vspace{8pt}

\noindent {\bfseries\sffamily Statement of Lemma~\ref{lem:enclose}:}\hspace{4pt}\emph{Consider $\vor(Y)$ for some point set $Y$.  For any point $x \in \mathbb{R}^2$, let $Q^*_x$ be the largest homothetic copy of $Q^*$ centered at $x$ such that $\Int(Q^*_x) \cap Y = \emptyset$.  Let $w_1$ and $w_2$ be two adjacent Voronoi edge bends or Voronoi vertices in $\vor(Y)$.   For any point $x \in w_1w_2$, $Q^*_x \subseteq Q^*_{w_1} \cup Q^*_{w_2}$.  The same property holds if $w_1$ and $w_2$ are Voronoi vertices connected by a Voronoi edge, and $x$ lies on that Voronoi edge.}
\begin{proof}
	We assume that $Q^*_x$ is not equal to $Q^*_{w_1}$ or $Q^*_{w_2}$ as there is nothing to prove otherwise.  Let $q$ and $q'$ be two of the defining points of $w_1$ and $w_2$.  So $w_1$ and $w_2$ lie on the Voronoi edge $e$ defined by $q$ and $q'$.  Place an imaginary point $q_1$ in $\partial Q^*_{w_1} \setminus Q^*_{w_{2}}$ such that $q_1$ does not lie on the same edge of $Q^*_{w_1}$ as $q$ or $q'$.  Place an imaginary point $q_{2}$ similarly in $\partial Q^*_{w_{2}} \setminus Q^*_{w_1}$.   Figure~\ref{fg:enclose-app}(a) shows an example.
	
	\begin{figure}
		\centering
		\begin{tabular}{ccc}
			\includegraphics[scale=0.5]{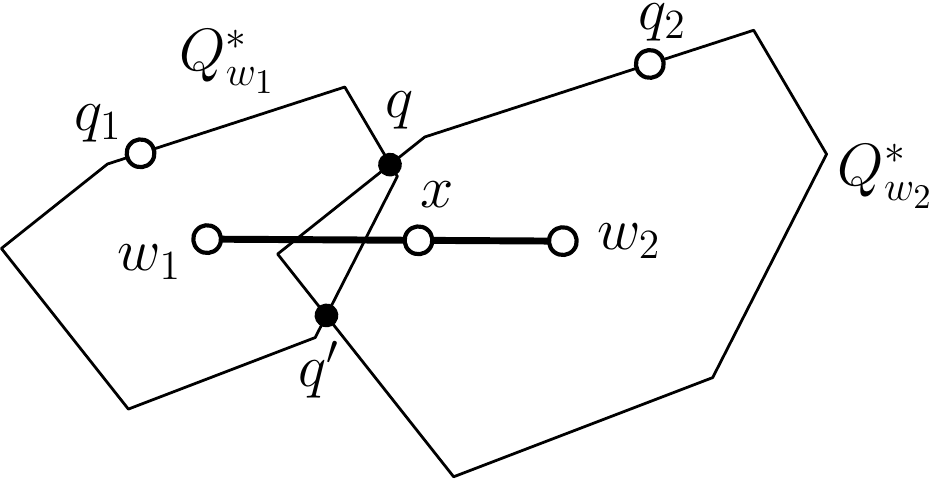} & & 
			\includegraphics[scale=0.5]{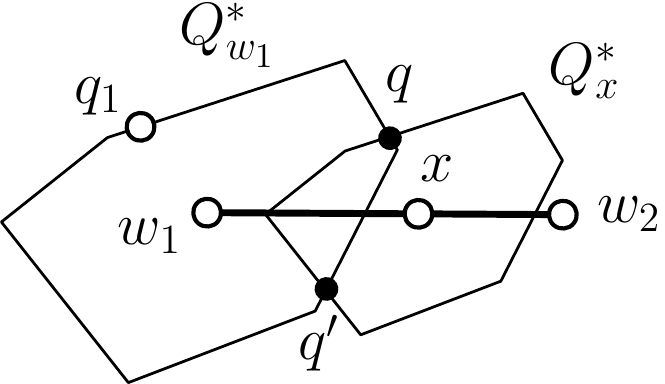} \\ \\
			(a) & \hspace*{.25in} & (b) \\
		\end{tabular}
		\caption{(a) The placement of $q_1$ and $q_2$ in $\partial Q^*_{w_1}$ and $\partial Q^*_{w_2}$. (b)~The chain in $\partial Q^*_{w_1}$ that goes from $q$ to $q'$ through $q_1$ lies outside $Q^*_x$.}
		\label{fg:enclose-app}
	\end{figure}
	
	We claim that $d_Q(q_1,x) \geq d_Q(q',x) = d_Q(q,x)$.  Suppose not.  Then $d_Q(q_1,x) < d_Q(q',x) = d_Q(q,x)$.  Move along the Voronoi edge $e$ from $x$ towards $w_2$.  We must reach some point $y$ before reaching $w_{2}$ such that $d_Q(q_1,y) = d_Q(q',y) = d_Q(q,y)$ because $q_1$ is farther from $w_{2}$ than $q$, $q'$, and $q_{2}$.  Let $Q^*_y$ be the homothetic copy of $Q^*$ centered at $y$ that includes $q$, $q'$, and $q_1$ in its boundary.  As $Q^*_{w_1} \not= Q^*_y$, either one is strictly contained in the other, or their boundaries intersect transversally at two points.  The former case is impossible as at least one of $q$, $q'$, and $q_1$ would lie in the interior of $Q^*_{w_1}$ or $Q^*_y$, an impossibility.  If $\partial Q^*_{w_1}$ and $\partial Q^*_y$ intersect transversally at two points, then one of $q$, $q'$ and $q_1$ would not lie in $\partial Q^*_{w_1}$ or $\partial Q^*_y$, an impossibility again.  This proves our claim that $d_Q(q_1,x) \geq d_Q(q',x) = d_Q(q,x)$.  
	
	Similarly, $d_Q(q_{2},x) \geq d_Q(q',x) = d_Q(q,x)$.
	
	Since both $d_Q(q_1,x)$ and $d_Q(q_2,x)$ are at least $d_Q(q,x) = d_Q(q',x)$, neither $q_1$ nor $q_{2}$ belongs to $\Int(Q^*_x)$.   Since $Q^*_x \not= Q^*_{w_1}$, either one of $Q^*_x$ and $Q^*_{w_1}$ is strictly contained in the other, or their boundaries intersect transversally.  The former case is impossible because one of $q_1$, $q$ and $q'$ would lie in the interior of $Q^*_x$ or $Q^*_{w_1}$.  In the second case, $\partial Q^*_x$ and $\partial Q^*_{w_1}$ must intersect transversally at $q$ and $q'$.  It follows that one of the two chains in $\partial Q^*_{w_1}$ delimited by $q$ and $q'$ lies outside $Q^*_x$.  Figure~\ref{fg:enclose-app}(b) shows an example.  Since $d_{Q^*}(x,q_1) = d_Q(q_1,x) \geq d_Q(q,x) = d_Q(q',x)$, we conclude that the chain that contains $q_1$ lies outside $Q^*_x$.  Similarly, we can show that $Q^*_x$ does not contain the chain in $\partial Q^*_{w_{2}}$ that goes from $q$ through $q_{2}$ to $q'$.  Hence, $Q^*_x$ must be contained in $Q^*_{w_1} \cup Q^*_{w_{2}}$.
	
	The same proof applies when $w_1$ and $w_2$ are Voronoi vertices connected by a Voronoi edge, and $x$ lies on that Voronoi edge.
\end{proof}

\vspace{8pt}

\noindent {\bfseries\sffamily Statement of Lemma~\ref{lem:conflict-tech}:}\hspace{4pt}\emph{Let $q$ be a point in some point set $Y$.  Let $quv$ be a triangle in the triangulated $\vor(Y)$.  If a point $p \not\in Y$ conflicts with a point in $quv$, then $p$ conflicts with $u$ or $v$.  Hence, if $p$ conflicts with $V_q(Y)$, $p$ conflicts with a Voronoi edge bend or Voronoi vertex in $\partial V_q(Y)$.}
\begin{proof}
	For any point $y \in \mathbb{R}^2$, let $Q^*_y$ be the largest homothetic copy of $Q^*$ centered at $y$ such that $\Int(Q^*_y) \cap Y = \emptyset$.  It suffices to show that the point $p$ that conflicts with $V_q(Y)$ belongs to $Q^*_u$ or $Q^*_v$.  If $p$ conflicts with any point in $uv$, Lemma~\ref{lem:enclose} implies that $p \in Q^*_u$ or $p \in Q^*_v$.  Suppose that $p$ does not conflict with any point in $uv$.  
	So $uv \subseteq V_q(Y \cup \{p\})$.
	Recall that the Voronoi cells of a Voronoi diagram under a convex distance function are star-shaped with respect to their sites.  So $V_q(Y \cup \{p\})$ is star-shaped with respect to $q$.   However, some segment that connects $q$ to some point in $uv$ must cross $V_p(Y \cup \{p\})$ as $p$ conflicts with a point in $quv$, a contradiction.   Figure~\ref{fg:conflict-app-1} illustrates this situation.
\end{proof}

\begin{figure}
	\centerline{\includegraphics[scale=0.5]{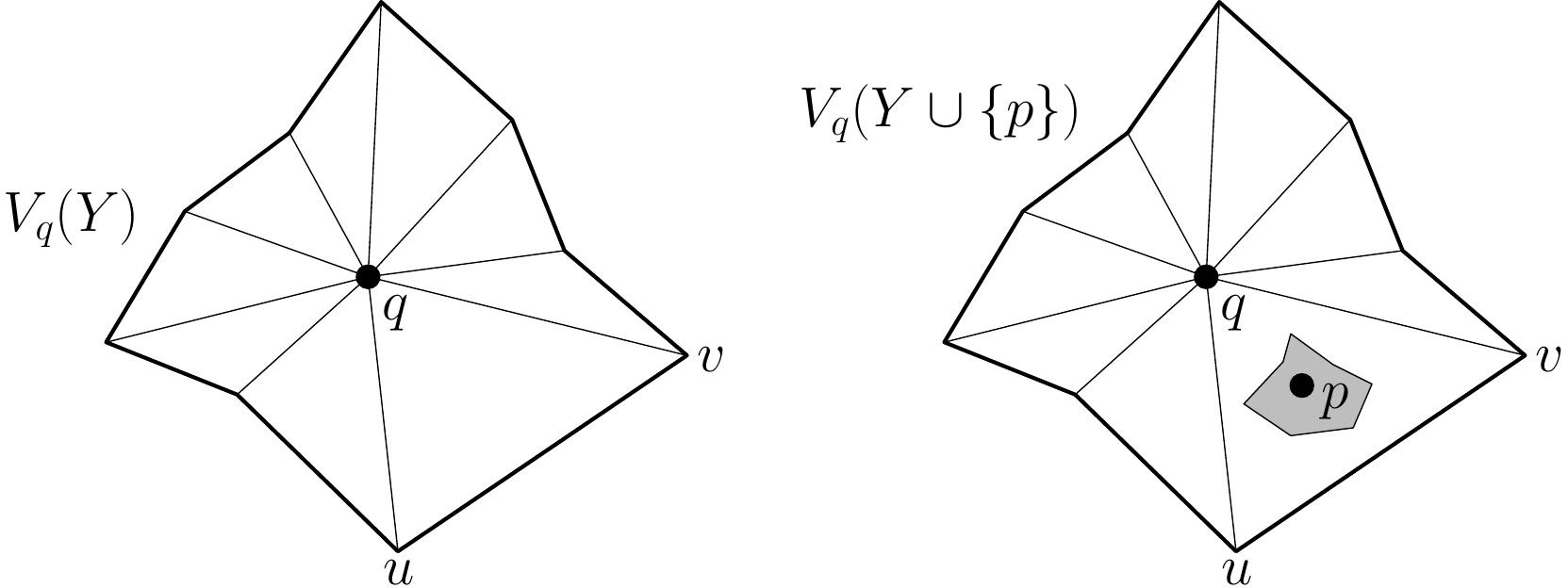}}
	\caption{The left image shows a schematic drawing of $V_q(Y)$.  If $V_p(Y \cup \{p\})$ lies inside $quv$ such as the shaded region on the right, then $V_q(Y \cup\{p\})$ is obtained by removing the shaded region from $V_q(Y)$, and $V_q(Y \cup \{p\})$ cannot be star-shaped with respect to $q$.}
	\label{fg:conflict-app-1}
\end{figure}

\cancel{
\begin{figure}
	\centerline{\includegraphics[scale=0.5]{figures/conflict-2}}
	\caption{The left image shows a schematic drawing of the situation and $Q^*_x$ (the dashed convex polygon).  The point $q$ is in $\partial Q^*_x$ as it is the closest site to $x$.  The point $p$ lies inside $Q^*_x$ as $p$ is closer to $x$ than $q$.  The right image shows $Q^*_u$ and $Q^*_x$.  Both contain $q$ in their boundaries.  Since they are homothetic copies of the same convex polygon, $Q^*_u$ must contain $Q^*_x$.}
	\label{fg:conflict-app-2}
\end{figure}
}

\section{Missing details in Section~\ref{sec:locate}}
\label{app:step1}

We will show that Task~1 in Section~\ref{sec:locate} runs in $O(n2^{O(\log^* n)})$ expected time and Task~2 in Section~\ref{sec:locate} runs in $O(n\log m)$ expected time.

\subsection{Lower envelope of two lower envelopes}
\label{app:step1-1}

We describe a algorithm based on the randomized divide-and-conquer approach due to Chan~\cite[Section~4]{chan16} to compute the lower envelope of ${\cal L}(B)$ and ${\cal L}(I)$.
	
Construct point location data structures for the triangulated $\vor(B)$ and $\vor(I)$ in $O(n)$ time.  Let $\mathtt{B}$ be the multiset version of $B$ in which each $p$ has multiplicity equal to the complexity of $V_p(B)$.  Draw a random sample $\mathtt{B}'$ of $\mathtt{B}$ of size $O(n/\log n)$, and let $B'$ denote the equivalent of $\mathtt{B}'$ without the multiplicity.  Similarly, $\mathtt{I}$ is multiset version of $I$, $\mathtt{I}'$ is a random sample of $\mathtt{I}$ of size $O(n/\log n)$, and $I'$ is the equivalent of $\mathtt{I}'$ without the multiplicity.

Compute ${\cal L}(B' \cup I')$, which has $O(n/\log n)$ size, directly in $O((n/\log n) \cdot \log n) = O(n)$ time.  For each triangle $t$ in ${\cal L}(B' \cup I')$, the strategy is to extract the subsets $B|_t = \{p \in B: \text{some point in $t$ is above $C_p$} \}$ and $I|_t = \{p \in I : \text{some point in $t$ is above $C_p$}\}$, and then recursively compute the patch ${\cal L}\bigl(B|_t\cup I|_t  \cup \{p_t\}\bigr) \cap \hat{t}$, where $p_t$ is the point in $B' \cup I'$ such that $t \subseteq C_{p_t}$, and $\hat{t}$ is the vertical prism obtained by sweeping $t$ upward and downward.   Collect these patches over all triangles in ${\cal L}(B' \cup I')$ and stitch them together to form ${\cal L}(B \cup I)$.  The recurrence for the expected running time is $T(n) = \sum_{ t \in {\cal L}(B' \cup I')} T(n_t+1) + E$, where $n_t = \bigl|B|_t \bigr| + \bigl|I|_t\bigr|$ and $E$ is the total expected running time to identify $B|_t$ and $I|_t$ for all triangles $t$ in ${\cal L}(B' \cup I')$.  

We identify $B|_t$ as follows.  
If some point $x$ of $t$ is above a cone $C_p$, then $p$ conflicts with the projection of $x$ in the projection of $t$.   Lemma~\ref{lem:conflict-tech} implies that $p$ conflicts with the projection of a vertex of $t$ different from $p_t$.
Take a vertex $v$ of $t$ different from $p_t$.  Locate $v$'s vertical projection in a triangle $t_v$ in the triangulated $\vor(B)$ by a point location query.  If $v$ is below $t_v$, then $v$ does not conflict with any point in $B$.  If $v$ is above $t_v$, we search $\vor(B)$ within the vertical projection of $t$ to determine the subset $B|_v = \{p \in B \setminus B' : \text{$v$ is above $C_p$}\}$.  The time needed is $O\bigl(\log n + \sum_{p \in B|_v} |\partial V_p(B)|\bigr)$.  Summing over all triangles in ${\cal L}(B' \cup I')$, the total point location time is $O(n/\log n \cdot \log n) = O(n)$.   By Clarkson-Shor's analysis~\cite[Corollary~3.8]{CS89}, it holds with probability at least $2/3$ that the sum of $\sum_{p \in B|_v} |\partial V_p(B)|$ over all vertices of ${\cal L}(B' \cup I')$ is $O(n)$, and $\max_{t \in {\cal L}(B' \cup I')} \max\{n_t\} = O(\log^2 n)$.  We identify $I|_t$ in the same way, which involves determining $I_v = \{p \in I \setminus I' : \text{$v$ is above $C_p$}\}$ for the vertices $v$ of ${\cal L}(B' \cup I')$, and the same analysis applies.  

When determining $B|_v$ and $I|_v$ over all vertices of ${\cal L}(B' \cup I')$, if the total number of steps exceeds $cn$ for some appropriate constant $c$, we abort, resample $\mathtt{B}'$ and $\mathtt{I}'$, and then repeat.  Similarly, if $\max_{t \in {\cal L}(B' \cup I')} \max\{n_t\}$ exceeds $c'\log^2 n$ for some appropriate constant $c'$, we also abort, resample $\mathtt{B}'$ and $\mathtt{I}'$, and then repeat.  We expect to succeed in $O(1)$ trials and proceed with the recursive calls.  

In summary, $E$ is $O(n)$ in the recurrence $T(n) = \sum_{ t \in {\cal L}(B' \cup I')} T(n_t+1) + E$, and there are $O(\log^* n)$ levels of recursion in expectation.  The hidden big-Oh constant may thus be raised to a power of $O(\log^* n)$, resulting in an expected running time of $n2^{O(\log^* n)}$.

\subsection{Determining $\pmb{t_1, t_2, \ldots,t_n}$}
\label{app:step1-2}

We generalize a method in~\cite{ailon11}, which does not work directly in our case because no information about $\vor(B)$ is gathered in the training phase.  
Our method works in two stages for each point $p_i \in I$ as follows.

\begin{itemize}
	
	\item Stage~1: Determine some subset $B_i$ that satisfies $B \subseteq B_i \subseteq S$, and compute a Voronoi edge bend or Voronoi vertex $v_i$ in $\vor(B_i)$ that conflicts with $p_i$ and is known to be in $V_S$ or not.
	
	\item Stage~2:  Use $v_i$ to find the triangle $t_i$ that contains $p_i$.  
	
\end{itemize}

\noindent We provide the details of these two stages for each input point in the following.  We discuss the second stage first because it is easier.  

\subsubsection{Stage 2}


If $v_i \in V_S$, we search $\vor(S)$ from $v_i$ to find $V_S|_{p_i}$ (i.e., the subset of $V_S$ that conflict with $p_i$), which by Lemma~\ref{lem:enclose} will also give the triangle in the triangulated $\vor(S)$ that contains $p_i$.  The time needed is $O\bigl(\bigl|V_S|_{p_i}\bigr|\bigr)$.  

Suppose that $v_i \not\in V_S$.  Then $v_i$ cannot be a region boundary vertex in the $m^2$-division of $\vor(S)$, so $v_i$ lies inside a region in the $m^2$-division, say $\Pi$.   We check whether $p_i$ conflicts with any boundary vertex of $\Pi$.  For each boundary vertex $w$ of $\Pi$, let $Q^*_w$ be the largest homothetic copy of $Q^*$ centered at $w$ such that $\Int(Q^*_w) \cap B = \emptyset$.  These $Q^*_w$'s form an arrangement of $O(m^2)$ complexity.  The point $p_i$ conflicts with $w$ if and only if $p_i \in Q^*_w$.  So we do a point location in the arrangement in $O(\log m)$ time to decide whether $p_i$ is contained in $Q^*_w$ for some boundary vertex $w$ of $\Pi$.   

If so, 
we can search $\vor(S)$ from $w$ as before to find the triangle in the triangulated $\vor(S)$ that contains $p_i$.   Suppose that $p_i \not\in Q^*_w$ for any boundary vertex $w$ of $\Pi$.  We claim that $p_i$ lies inside $\Pi$, which means that we can do a point location in $\vor(S) \cap \Pi$ to find the triangle in the triangulated $\vor(S)$ that contains $p_i$.  
The time needed is $O(\log m)$ as $\vor(S) \cap \Pi$ has $O(m^2)$ size.    The running time of Stage~2 is $O\bigl( \bigl|V_S|_{p_i}\bigr| + \log m \bigr)$.

We prove the claim as follows.  Assume to the contrary that it is false.  Since $p_i$ conflicts with $v_i$ that lies inside $\Pi$, we have $p_i \in Q^*_{v_i}$, the largest homothetic copy of $Q^*$ centered at $v_i$ such that $\Int(Q^*_{v_i}) \cap B = \emptyset$.  The segment $p_iv_i$ intersects the boundary of $\Pi$ at some point $x$.   We can define a deformation of $Q^*_{v_i}$ so that its center moves linearly from $v_i$ to $p_i$ while the polygon shrinks linearly to the point $p_i$.   Since $x \in p_iv_i$, at some point during this deformation of $Q^*_{v_i}$, we must obtain a homothetic copy $\tilde{Q}_x$ of $Q^*_x$ such that $x$ is the center of $\tilde{Q}_x$, $\Int(\tilde{Q}_x) \cap B = \emptyset$, and $p_i \in \tilde{Q}_x$.  Then, we can invoke Lemma~\ref{lem:enclose} to obtain the contradiction that $p_i$ must conflict with a boundary vertex of $\Pi$.

\subsubsection{Stage~1}  


For efficiency purpose, we will present a procedure that runs stage~1 for all input points in $I$ in an inductive manner.  The procedure is a generalization of a method for a similar task in~\cite{ailon11}.  During the running of this procedure, whenever we have computed $v_i$ for an input point $p_i$ as required of stage~1, the procedure will invoke stage~2 for $p_i$ in order to locate the triangle $t_i$, and if $v_i \not\in V_S$, compute a Voronoi edge bend or Voronoi vertex $v'_i \in V_S$ that conflicts with $p_i$.

We first present a technical result that will be used by the procedure.

\begin{lemma}
	\label{lem:cyclic}
	Let $p$ be a point in $B \cup I$.  Let $B_p$ be any subset that satisfies $B \subseteq B_p \subseteq S$.  Assume that $V_p(B_p \cup I)$, $V_p(B_p \cup \{p\})$, and the edges of $\vor(B_p)$ that intersect $V_p(B_p \cup \{p\})$  are known.  For each $p_i \in N_p(B_p \cup I) \setminus N_p(B_p \cup \{p\})$, we can compute a 
	Voronoi edge $e_i$ in $V_p(B_p)$ that conflicts with $p_i$.  The total running time is $O\bigl(\bigl|N_p(B_p \cup I)\bigr| + \bigl|N_p(B_p \cup \{p\})\bigr|\bigr)$.
\end{lemma}
\begin{proof}
	Suppose that $p \in I$.  Take a point $p_i \in N_p(B_p \cup I) \setminus N_p(B_p\cup \{p\})$.  The segment $pp_i$ lies between $pq$ and $pq'$ for some $q, q' \in N_p(B_p \cup \{p\})$ that are consecutive in the cyclic order of $N_p(B_p \cup \{p\})$ around $p$.   
	Recall that the dummy points make all Voronoi cells of input points bounded.  So there is a Voronoi vertex $w$ in $\partial V_p(B_p \cup \{p\})$ defined by $p$, $q$ and $q'$.  
	
	We claim that $p_i$ conflicts with $w$.  Let $Q^*_w$ be the largest homothetic copy of $Q^*$ centered at $w$ that circumscribes $p$, $q$ and $q'$.  Since $p_i \in N_p(B_p \cup I)$, there exists a homothetic copy $Q^*_x$ of $Q^*$ such that $\{p,p_i\} \subset \partial Q^*_x$ and $\Int(Q^*_x) \cap (B_p \cup I) = \emptyset$.  If $p_i$ does not conflict with $w$, then $p_i \not\in Q^*_w$.  Refer to Figure~\ref{fg:cyclic}.  But then as $Q^*_w$ and $Q^*_x$ intersects transversally at zero or two points, $Q^*_x$ is forced to contain $q$ or $q'$ in its interior, a contradiction.
	%
	%
	Clearly, $q$ and $q'$ define a Voronoi edge $e_i$ in $\vor(B_p)$ that intersects $V_p(B_p \cup \{p\})$ and contains $w$, so $e_i$ is the Voronoi edge that we look for.
	
	By the analysis above, a synchronized cyclic scan of $N_p(B_p \cup \{p\})$ and $N_p(B_p \cup I)$ gives the Voronoi edges of $V_p(B_p)$ that we look for.
	
	The remaining case is that $p \in B$.  Hence, $V_p(B_p \cup \{p\}) = V_p(B_p)$.  Every point $p_i \in N_p(B_p \cup I)$ must conflict with some point in $\partial V_p(B_p)$ in order that $p_i$ becomes a Voronoi neighbor of $p$ in $N_p(B_p \cup I)$.  Each $p_i \in N_p(B_p \cup I)$ conflicts with a connected portion of $\partial V_p(B_p)$.  Moreover, this portion of $\partial V_p(B_p)$ is not nested within the portion of $\partial V_p(B_p)$ that conflicts with any other $p_j \in N_p(B_p \cup I)$.  It follows that a synchronized cyclic scan of $\partial V_p(B_p)$ and $N_p(B_p \cup I)$ gives the Voronoi edges that we look for.
	%
\end{proof}

\begin{figure}
	\centerline{\includegraphics[scale=0.6]{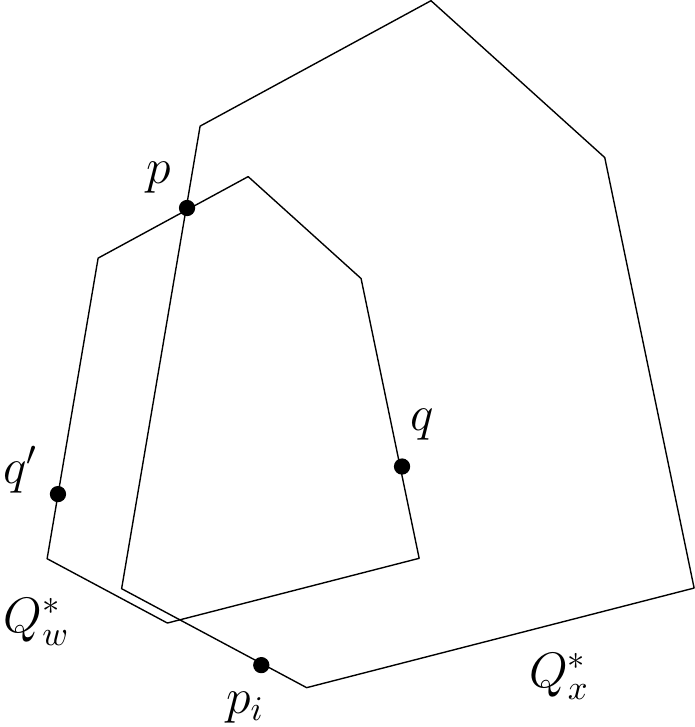}}
	\caption{Illustration for the proof of Lemma~\ref{lem:cyclic}.}
	\label{fg:cyclic}
\end{figure}

\cancel{
\begin{figure}
	\centerline{\includegraphics[scale=0.4]{ill}}
	\caption{The two Voronoi edges defined by $p$ and $q$ and $q'$, respectively, extend downward to infinity.}
	\label{fg:ill}
\end{figure}
}

\cancel{
Next, we show that once $w_i$ is found by Lemma~\ref{lem:cyclic} for a point $p_i \in N_p(B_p \cup I) \setminus N_p(B_p \cup \{p\})$, we can determine in $O(1)$ time a Voronoi edge bend or Voronoi vertex $v_i$ in $\vor(B_p)$ that conflicts with $p_i$.

\begin{lemma}
	\label{lem:vi}
	Given the setting and $w_i$ in Lemma~\ref{lem:cyclic}, we can determine in $O(1)$ time a Voronoi edge bend or Voronoi vertex $v_i$ in $\vor(B_p)$ that conflicts with $p_i$.  
\end{lemma}
\begin{proof}
	Suppose that $w_i$ is a Voronoi vertex defined by $p$ and $q, q' \in N_p(B_p \cup \{p|)$.  In $\vor(B_p)$, $w_i$ lies on the Voronoi edge  $e$ between $q$ and $q'$.  Observe that $p_i$ still conflicts with $w_i$ in $\vor(B_p)$.  By Lemma~\ref{lem:enclose}, among the two Voronoi edge bends or Voronoi vertices in $e$ that are adjacent to $w_i$, $p_i$ must conflict with at least one of them.
	
	The other case is that the two Voronoi edges in $\partial V_p(B_p \cup \{p\})$ induced by $p$ with $q$ and $q'$ are infinite, and $w_i$ is a Voronoi edge bend on the infinite Voronoi edge between $p$ and $q$.

\end{proof}
}

The following is the pseudocode for determining the triangles in the triangulated $\vor(S)$ that contains the input points in $I$.

\begin{quote}
	\begin{enumerate}
		\item Initialize a queue $L$ to contain all points in $B$.
		\item Mark all points in $B \cup I$ as unvisited.
		\item While $L$ is non-empty do:
		\begin{enumerate}[(\alph{enumii})]
			\item Dequeue the next point $p$ from $L$.
			\item If $p \in B$, let $B_p = B$.  Otherwise, $p = p_j \in I$, and $v_j$ has been inductively determined, and we perform the following steps:
			\begin{enumerate}[(\roman{enumiii})]
				\item We will show in Lemma~\ref{lem:vi2} below that $v_j \in V_S$.  Search $\vor(S)$ from $v_j$ to determine $V_S|_{p}$.
				\item Let $S_{p}$ be the set of defining points of the elements of $V_S|_{p}$.  Let $S'_{p}$ be the set of defining points of Voronoi edge bends and Voronoi vertices in $\vor(S)$ that are adjacent to the elements of $V_S|_{p}$.  Note that $|S_{p}|$ and $|S'_{p}|$ are $O\bigl(\bigl|V_S|_{p}\bigr|\bigr)$.
				\item $B_p := B \cup S_{p} \cup S'_{p}$.   The motivation for this definition of $B_p$ is to ensure that $V_p(B_p \cup \{p\}) = V_p(S \cup \{p\})$.
				\item In the same search in step~3(b)(i), we construct $V_p(B_p \cup \{p\})$ and the edges of $\vor(B_p)$ that intersect $V_p(B_p \cup \{p\})$ without increasing the asymptotic running time.
				\item Merge $V_p(B_p \cup \{p\})$ and $V_p(B \cup I)$ to form $V_p(B_p \cup I)$.
			\end{enumerate}
			\item Use Lemma~\ref{lem:cyclic} to determine for each $p_i \in N_p(B_p \cup I) \setminus N_p(B_p \cup \{p\})$, the 
			Voronoi edge $e_i$ in $\vor(B_p)$ that conflicts with $p_i$.  By Lemma~\ref{lem:enclose}, $p_i$ must conflict with some Voronoi edge bend or endpoint of $e_i$, which is the desired Voronoi edge bend or Voronoi vertex $v_i$ in $\vor(B_p)$ for each unvisited $p_i \in N_p(B_p \cup I) \setminus N_p(B_p \cup \{p\})$.
			\item For each unvisited $p_i \in N_p(B_p \cup I) \setminus N_p(B_p \cup \{p\})$, 
			\begin{enumerate}[(\roman{enumiii})]
				\item Invoke the second stage to find the triangle $t_i$ in the triangulated $\vor(S)$ that contains $p_i$, mark $p_i$ as visited, and append $p_i$ to $L$.
				\item If $p \in B$, let $u, u' \in V_S$ be two of the vertices of $t_i$, and update $v_i$ to be $u$ or $u'$ whichever conflicts with $p_i$.  Note that if $p \in I$, we will prove in Lemma~\ref{lem:vi2} below that $v_i$ already belongs to $V_S$.
			\end{enumerate}
		\end{enumerate}
	\end{enumerate}
\end{quote}

\begin{lemma}
	\label{lem:vi2}
	At the end of step~3(c), if $p \in I$, then for each $p_i \in N_p(B_p \cup I) \setminus N_p(B_p \cup \{p\})$, $v_i \in V_S$.  At the end of step~3(d)(ii), for each visited $p_i \in I$, $v_i \in V_S$.
\end{lemma}
\begin{proof}
	We prove the lemma by induction.  Consider an iteration of the while-loop in step~3.  The newly visited points $p_i \in I$ are those in $N_p(B_p \cup I) \setminus N_p(B_p \cup \{p\})$.  
	
	Suppose that $p \in I$.  By the definition of $B_p$, $V_p(B_p \cup \{p\}) = V_p(S \cup \{p\})$.  Therefore, 
	$e_i$ in Lemma~\ref{lem:cyclic} is a Voronoi edge in $\vor(S)$ that conflicts with $p_i$.   Moreover, the proof of Lemma~\ref{lem:cyclic} reveals that $p_i$ conflicts with a Voronoi vertex $w$ of $\vor(B_p \cup \{p\})$ that lies on $e_i$, and $w$ is defined by $p$ and the two defining points of $e_i$.  Therefore, $p$ conflicts with $w$ on the edge $e_i$ of $\vor(S)$ too.  By Lemma~\ref{lem:enclose}, $p$ conflicts with some point in $V_S \cap e_i$.
	The defining points of $V_S|_p$ are included in $B_p$ by definition, which includes the defining points of $V_S|_p \cap e_i$.  Therefore, we have obtained the edge $e_i$ during the construction of $V_p(B_p \cup \{p\})$, which allows $e_i$ to be returned for $p_i$ in the application of Lemma~\ref{lem:cyclic} in step~3(c).  When we search along $e_i$ in step~3(c), we find the point $v_i \in V_S$ that conflicts with $p_i$.
	
	Suppose that $p \in B$.  Then, $B_p = B$.  Consider step~3(d)(ii).  Let the vertices of $t_i$ be $\{q,u,u'\}$, where $q$ is a point in $S$.  Since $p_i \in t_i$, 
	Lemma~\ref{lem:conflict-tech} implies that $p_i$ conflicts with $u$ or $u'$.  Note that both $u$ and $u'$ belong to $V_S$.
\end{proof}
	
The correctness of the pseudocode follows from induction using Lemmas~\ref{lem:cyclic} and~\ref{lem:vi2}.

We can view step~3(b)(v) as taking the lower envelope of two polygonal cones.  In particular, we lift $V_p(B_p \cup \{p\})$ and $V_p(B \cup I)$ to $\mathbb{R}^3$.  Then $V_p(B_p \cup I)$ is the lower envelope of these two polygonal cones, which can be obtained by a synchronized cyclic scan of $N_p(B_p \cup \{p\})$ and $N_p(B \cup I)$ in linear time.  In step~3(d)(i), when we invoke stage~2 for a point $p_i \in I$, we are supposed to know whether $v_i \in S$.  If $p \in I$, then Lemma~\ref{lem:vi2} ensures that $v_i \in S$.  If $p \in B$, we can assume that all Voronoi edge bends and Voronoi vertices in $\vor(B)$ have been labelled in preprocessing whether they belong to $V_S$ or not.

The total running time is $O\bigl(\sum_{p \in B \cup I} |N_p(B \cup I)| + \sum_{p \in B} |N_p(B)| + \sum_{p_i \in I}\bigl|V_S|_{p_i}\bigr| + n\log m \bigr)$ from our previous discussion.  The first two terms are $O(n)$.  By Lemma~\ref{lem:R}, the expected value of the third term is $O(n)$.  Therefore, the expected running time of the pseudocode is $O(n\log m)$.

\section{Missing details in Section~\ref{sec:kNN}}
\label{app:NN}

\subsection{$\pmb c$-WSPD from a compression $\pmb{T_X}$ of $\pmb{T_S}$}
\label{app:NN-1}

Let $T_X$ be the compression of $T_S$ to $X$.   We compute a $c$-WSPD as described in~\cite{har-peled06} for a doubling metric, which is $d$ in our case, using $T_X$.  The pseudocode is given below.  The top-level call is Build$(r,r)$, where $r$ is the root of $T_X$.
\begin{quote}
	Build$(u,v)$
	\begin{enumerate}
		\item Swap $u$ and $v$ if necessary to ensure that $\ell_u > \ell_v$, or $\ell_u = \ell_v$ and $u$ is to the left of $v$ in the inorder traversal.
		\item If $\frac{4\tau}{\tau-1} \cdot \tau^{\ell_u} < \frac{1}{c+2} \cdot d(p_u,p_v)$ then return $\bigl\{\{P_u,P_v\}\bigr\}$.
		\item Otherwise, let $w_1,\ldots,w_j$ be the children of $u$, return $\bigcup_{i=1}^j \mathrm{Build}(w_i,v)$.
	\end{enumerate}
\end{quote}

\begin{lemma}
	\label{lem:build}
	\emph{Build} constructs a $c$-WSPD of size $O((c+1)^{O(1)}|X|)$  in $O((c+1)^{O(1)}|X|)$ time.
\end{lemma}
\begin{proof}
	We first show that Build outputs a $c$-WSPD.  Clearly, the working of Build guarantees that for every distinct pair of points $x, y \in X$, there exists a pair $\{P_u,P_v\}$ in the output of Build such that $x \in P_u$ and $y \in P_v$.  Let $\{P_u,P_v\}$ be a pair in the output of Build.  Let $P'_u$ and $P'_v$ be the subsets of points for the subtrees of $T_S$ rooted at $u$ and $v$.  Let $\delta'_u$ and $\delta'_v$ be the diameters of $P'_u$ and $P'_v$ under $d$, respectively.  By property~(c) of a net-tree, $P'_u \subseteq B(p_u,\frac{2\tau}{\tau-1}\tau^{\ell_u})$.  Thus, $\max\{\delta'_u,\delta'_v\} \leq \frac{4\tau}{\tau-1}\max\{\tau^{\ell_u},\tau^{\ell_v}\}$, which is less than $\frac{1}{c+2}d(p_u,p_v)$ by steps~1 and~2 of Build.  Then, $\max\{\delta'_u,\delta'_v\} < \frac{1}{c+2}d(p_u,p_v) \leq \frac{1}{c+2}\bigl(d(P'_u,P'_v) + \delta'_u + \delta'_v\bigr)$.  Rearranging terms gives $\max\{\delta'_u,\delta'_v\} < \frac{1}{c} \cdot d(P'_u,P'_v)$.  Let $\delta_u$ and $\delta_v$ be the diameters of $P_u$ and $P_v$ respectively.  Observe that $P_u \subseteq P'_u$ and $P_v \subseteq P'_v$.  Hence, $\max\{\delta_u,\delta_v\}  \leq \max\{\delta'_u,\delta'_v\} < \frac{1}{c} \cdot d(P'_u,P'_v) \leq \frac{1}{c} \cdot d(P_u,P_v)$.  In summary, the output of Build is a $c$-WSPD.
	
	It remains to bound the output size and the running time of Build.
	
	Consider a pair $\{P_u,P_v\}$ in the output.  Without loss of generality, assume that Build$(u,v)$ is called by Build$(u,w)$, where $w = \mathit{parent}(v)$.  We charge the pair $\{P_u,P_v\}$ to $w$.  Since Build considers the children of $w$ instead of those of $u$ in processing $(u,w)$, we must have $\ell_w \geq \ell_u$.  We claim that $\ell_{\mathit{parent}(u)} \geq \ell_w$.  Suppose not.  There must be a call Build$(\mathit{parent}(u),w')$ before the call Build$(u,w)$, where $w' = w$ or $w'$ is an ancestor of $w$.  Since $\ell_{\mathit{parent}(u)} <\ell_w \leq \ell_{w'}$, Build$(\mathit{parent}(u),w')$ eventually leads to the call Build$(\mathit{parent}(u),v)$, which must then call Build$(u,v)$ for $\{P_u,P_v\}$ to be included in the output.  But this contradicts our assumption that Build$(u,v)$ is called by Build$(u,w)$.  Therefore, $\ell_{\mathit{parent}(u)} \geq \ell_w \geq \ell_u$.  Hence, $u$ belongs to the set $N_X(\ell_w)$ as defined for $T_X$ below:
	\[
	N_X(\ell) = \{\text{node $u$ of $T_X$: $\ell_u \leq \ell \leq \ell_{\mathit{parent}(u)}$}\}.
	\]
	Note that the parent-child relation in the definition of $N_X(\ell)$ refers to $T_X$.   Consider the following similar definition for $T_S$:
	\[
	N_S(\ell) = \{\text{node $u$ of $T_S$: $\ell_u \leq \ell \leq \ell_{\mathit{parent}(u)}$}\}.
	\]
	Note that the parent-child relation in the definition of $N_S(\ell)$ refers to $T_S$.   Since the pair $\{P_u,P_w\}$ is not included in the output, we must have $\frac{4\tau}{\tau-1}\tau^{\ell_w} \geq \frac{1}{c+2} d(p_u,p_w)$.  That is, $p_u \in B\bigl(p_w,O(c\tau^{\ell_w})\bigr)$.   We analyze the total charge on $w$ in the following.
	
	First, we charge the nodes in $N_X(\ell_w)$ to some nodes in $N_S(\ell_w)$.  Take any node $u \in N_X(\ell_w)$.  If $u \in N_S(\ell_w)$, we charge $u$ to itself.  Suppose that $u \not\in N_S(\ell_w)$.  Let $u'' = \mathit{parent}(u)$ in $T_X$.  It means that there are some internal nodes on the path from $u''$ to $u$ in $T_S$, and all of them are pruned by the compression of $T_S$ to $X$.  It follows from the definition of $N_S(\ell)$ that exactly one of these pruned internal node belongs to $N_S(\ell_w)$, say $u'$.  We charge $u \in N_X(\ell_w)$ to $u' \in N_S(\ell_w)$.   Note that $u'$ cannot be charged by another node in $N_X(\ell_w)$.  Otherwise, $T_X$ would contain nodes in two different child subtrees of $u'$ in $T_S$, which would force $u'$ to be a node of $T_X$.  This contradicts the fact that $u'$ is pruned by the compression of $T_S$ to $X$.
	
	Second, consider a node $u \in N_X(\ell_w)$ that charges an ancestor $u' \in N_S(\ell_w)$ of $u$ in $T_S$.  We have shown earlier that $d(p_u,p_w) = O(c\tau^{\ell_w})$ if $\{P_u,P_w\}$ does not appear in the output of Build.  By property~(d) of the net-tree $T_S$, we have $d(p_u,p_{u'}) = O(\tau^{\ell_{u'}})$, which is $O(\tau^{\ell_w})$ as $\ell_{u'} \leq \ell_w$ by the definition of $N_S(\ell_w)$.  As a result, $d(p_{u'}, p_w) \leq d(p_u,p_{u'}) + d(p_{u}, p_w) = O((c+1)\tau^{\ell_w})$.  Consequently, the nodes in $N_S(\ell_w)$ that are charged by the nodes in $\bigl\{u \in N_X(\ell_w) : \text{$\{P_u,P_w\}$ does not appear in the output of Build} \bigr\}$ lie in $B(p_w,O((c+1)\tau^{\ell_w}))$.
	
	It is known that any two nodes in $N_S(\ell_w)$ are at distance $\frac{1}{4}\tau^{\ell_w-1}$ or more apart~\cite[Proposition~2.2]{har-peled06}.  Therefore, $N_S(\ell_w)\cap B(p_w,O((c+1)\tau^{\ell_w}))$ has size at most $(c+1)^{O(1)}$ by the doubling property.  
	
	In summary, the total charge on the node $w$ in $T_X$ is $(c+1)^{O(1)}$.  Since $T_X$ has $O(|X|)$ nodes, the size bound of the $c$-WSPD follows.

	Construct a computation tree $\cal T$ in which each node is labelled $(u,v)$ for the call Build$(u,v)$, and a node $(u,v)$ is a child of another node $(u,w)$ if Build$(u,w)$ calls Build$(u,v)$.  The leaves of $\cal T$ correspond to the pairs output by Build.  Each internal node of $\cal T$ has at least two children because each internal node of $T_X$ has at least two children. So $\cal T$ has $O((c+1)^{O(1)}|X|)$ nodes as it has $O((c+1)^{O(1)}|X|)$ leaves.  Clearly, Build spends $O(1)$ time at each node of $\cal T$, establishing the running time bound.
\end{proof}

\subsection{Correctness of the extraction of $\pmb{k}$-nearest neighbor graph}
\label{app:NN-2}


Compute a subset $C_v \subseteq X$ for every leaf $v$ of $T_X$ such that $|C_v| = O(k)$ and $C_v$ contains the subset $\bigl\{p \in X : \text{the point in $P_v$ is a $k$-nearest neighbor of $p$}\bigr\}$.  The containment may be strict, so the point in $P_v$ may not be a $k$-nearest neighbor of some point $p \in C_v$.  We will discuss shortly how to compute such $C_v$'s.  After obtaining all $C_v$'s, for each point $p \in X$, construct $L_p = \bigcup \bigl\{ P_v : \text{$v$ is a leaf of $T_X$} \wedge p \in C_v \bigr\}$.   By definition, all $k$-nearest neighbors of $p$ are included in $L_p$, although $L_p$ may contain more points.    We select in $O(|L_p|)$ time the $k$-th farthest point $p'$ in $L_p$ from $p$ under $d$.  Then, we scan in $O(|L_p|)$ time using $d(p,p')$ to find the $k$-nearest neighbors of $p$.  Hence, as $|C_v| = O(k)$, the total running time is $O(\sum_{p \in X} |L_p|) = O(\sum_{\text{leaf $v$}} |C_v|) = O(k|X|)$ plus the time to compute the $C_v$'s.  It remains to discuss the computation of the $C_v$'s.

Compute a 4-WSPD $\Delta$ of $X$ which takes $O(|X|)$ time by Lemma~\ref{lem:build}.  We define a subset $C_u \subseteq X$ for every node $u$ of $T_X$ with a generalized requirement.   For every node $u$ of $T_X$, we require that $|C_u| = O(k)$ and $C_u$ contains the subset $\bigl\{p : \exists \, \{P_w,P_{w'}\} \in \Delta \, \text{s.t.} \, p \in P_{w'}$, $w$ is $u$ or an ancestor of $u$, and $P_u$ contains a $k$-nearest neighbor of $p$ $\bigr\}$.  The containment may be strict, i.e., some $p \in C_u$ may violate the property above.

This generalization is consistent with the requirement for $C_v$ at a leaf $v$ because if the single point $q$ in $P_v$ is a $k$-nearest neighbor of some $p$, then as $\Delta$ is a WSPD, there exists $\{P_w,P_{w'}\} \in \Delta$ such that $q \in P_w$ and $p \in P_{w'}$; $w$ is clearly either $v$ or an ancestor of $v$.

The $C_u$'s are generated in a preorder traversal of $T_X$.   The computation of $C_u$ is complete after visiting $u$.  When visiting a node $u$, we initialize a set $C$ and prune $C$ later to obtain $C_u$.  The initial $C$ is $C_{\mathit{parent}(u)} \cup \bigl\{p : \exists \, \{P_u,P_{w'}\} \in \Delta \; \text{s.t.} \; p \in P_{w'} \wedge |P_{w'}| \leq k \bigr\}$.  (If $u$ is the root of $T$, take $C_{\mathit{parent}(u)}$ to be $\emptyset$.)    Lemma~\ref{lem:initial} below shows that the initial $C$ satisfies the requirement for $C_u$ except that $|C_u|$ may not be $O(k)$.   As $|C_{\mathit{parent}(u)}| = O(k)$ inductively, the initialization of $C$ takes $O\bigl(k + \bigl|\bigl\{p : \exists \, \{P_u,P_{w'}\} \in \Delta \; \text{s.t.} \; p \in P_{w'} \wedge |P_{w'}| \leq k\bigr\}\bigr|\bigr)$ time.  

\begin{lemma}
	\label{lem:initial}
	The initial $C$ contains the subset $\bigl\{p : \exists \, \{P_w,P_{w'}\} \in \Delta \, \text{s.t.} \, p \in P_{w'}$, $w$ is $u$ or an ancestor of $u$, and $P_u$ contains a $k$-nearest neighbor of $p$ $\bigr\}$.
\end{lemma}
\begin{proof}
	Let $K = \bigl\{p : \exists \, \{P_w,P_{w'}\} \in \Delta \, \text{s.t.} \, p \in P_{w'}$, $w$ is $u$ or an ancestor of $u$, and $P_u$ contains a $k$-nearest neighbor of $p$ $\bigr\}$.  Partition $K$ into a disjoint union $K' \cup K''$, where $K'$ covers those pairs $\{P_w,P_{w'}\} \in \Delta$ such that $w$ is an ancestor of $u$, and $K''$ covers those pairs $\{P_u,P_{w'}\} \in \Delta$.  Inductively, $K' \subseteq C_{\mathit{parent}(u)}$.  We just need to argue that $K''$ is contained in the subset $\bigl\{p : \exists \, \{P_u,P_{w'}\} \in \Delta \; \text{s.t.} \; p \in P_{w'} \wedge |P_{w'}| \leq k\bigr\}$ which is part of the initial $C$.  That is, if $p \in P_{w'}$ for some $\{P_u,P_{w'}\} \in \Delta$ and some $q \in P_u$ is a $k$-nearest neighbor of $p$, we need to show that $|P_{w'}| \leq k$.   In this case, as $\Delta$ is a 4-WSPD, the diameter of $P_{w'}$ is less than $\frac{1}{4}d(P_u,P_{w'}) < d(p,q)$.   So all points in $P_{w'} \setminus \{p\}$ are closer to $p$ than $q$, which implies that $|P_{w'}| \leq k$ because $q$ is a $k$-nearest neighbor of $p$. 
\end{proof}

We prune $C$ as follows.  Recall that $\hat{Q}$ is the polygon of $O(1)$ size that induces the metric $d$.   Let $\Xi = (\xi_1,\xi_2,\ldots)$ be a maximal set of points in $\partial \hat{Q}$ in clockwise order such that for any $\xi_i \in \Xi$, $d(\xi_i,\xi_{i+1}) \in \bigl[\frac{1}{8},\frac{1}{4}\bigr]$.   The set $\Xi$ has $O(1)$ size and can be computed in $O(1)$ time by placing points greedily in $\partial \hat{Q}$.  Let $\gamma_i$ be the ray from the origin through $\xi_i$.  These rays divide $\mathbb{R}^2$ into cones.  Fix an arbitrary point $q_0 \in P_u$.   Compute in $O(|C|)$ time, for all $i$, the subset $C_i$ of $C$ in the cone bounded by $\gamma_i + q_0$ and $\gamma_{i+1} + q_0$.  Determine the $k$-th nearest point in $C_i$ from $q_0$ in $O(|C_i|)$ time.  Then, scan $C_i$ in $O(|C_i|)$ time to retain only the $k$ nearest points in $C_i$ from $q_0$.  Repeat the same for every $C_i$.   The union of the pruned $C_i$'s is $C_u$ which clearly has $O(k)$ size.  Lemma~\ref{lem:kn2} below shows that the pruning of $C_i$ only removes a point $p$ if no point in $P_u$ can be a $k$-nearest neighbor of $p$.  Therefore, the union of the pruned $C_i$'s satisfies the requirement for $C_u$.  The running time over all $C_i$'s is $O\bigl(|C|) = O\bigl(k + \bigl|\bigl\{p : \exists \, \{P_u,P_{w'}\} \in \Delta \; \text{s.t.} \; p \in P_{w'} \wedge |P_{w'}| \leq k\bigr\}\bigr|\bigr)$.

In summary, the computation of all $C_u$'s takes $O(k|\Delta|) = O(k|X|)$ time.

\begin{lemma}
	\label{lem:kn2}
	For any point $p \in C_i$, if there are at least $k$ points in $p' \in C_i \setminus \{p\}$ such that $d(p',q_0) \leq d(p,q_0)$, then no point in $P_u$ can be a $k$-nearest neighbor of $p$.
\end{lemma}
\begin{proof}
		By our initialization of $C$, we can inductively show that for any point $p$ included in the initial $C$, there exists $\{P_w,P_{w'}\} \in \Delta$ such that $p \in P_{w'}$, and $w$ is $u$ or an ancestor of $u$.    
		
		Pick a point $p \in C_i$.  Let $p'$ be a point in $C_i \setminus \{p\}$ such that $d(p',q_0) \leq d(p,q_0)$.   Let $\lambda_0$ be the factor such that $p' \in \partial (\lambda_0\hat{Q} + q_0)$.  Similarly, let $\lambda_1 \geq \lambda_0$ be the factor such that $p \in \partial (\lambda_1\hat{Q} + q_0)$.  Let $p''$ be the intersection between $pq_0$ and $\partial (\lambda_0\hat{Q} + q_0)$.   Refer to Figure~\ref{fg:knn}.
		
		\begin{figure}
			\centerline{\includegraphics[scale=0.6]{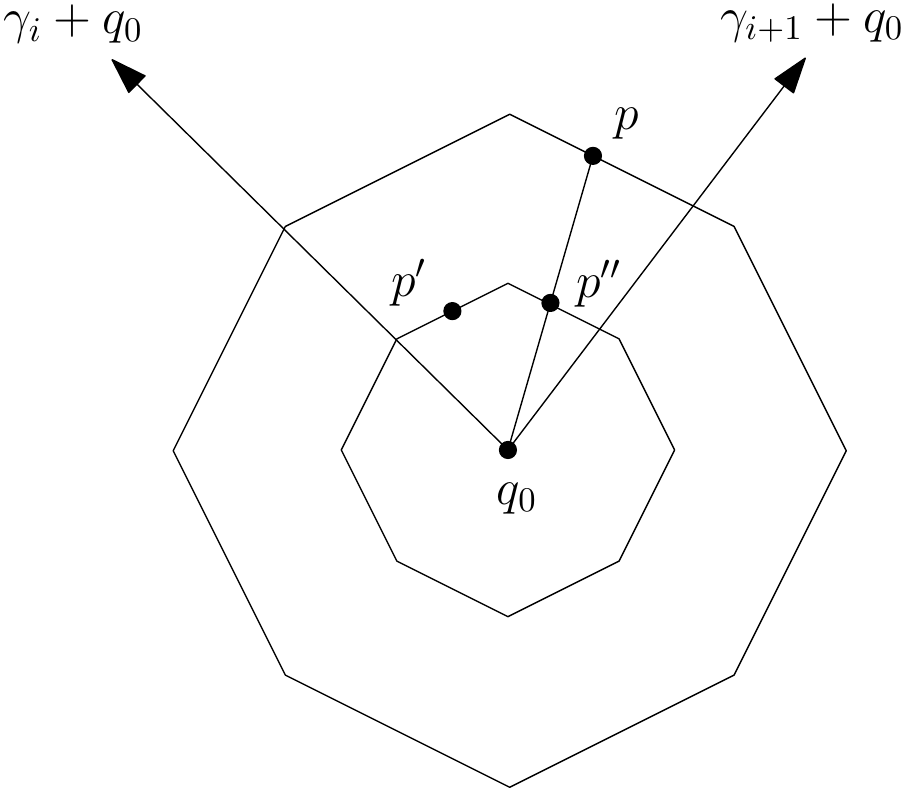}}
			\caption{The two concentric polygons are homothetic copies of $\hat{Q}$.  The smaller one is $\lambda_0 \hat{Q}+q_0$.  The larger one is $\lambda_1\hat{Q} + q_0$.}
			\label{fg:knn}
		\end{figure}
		
		Since $\{p',p''\} \subset \partial(\lambda_0 \hat{Q} + q_0)$, and $p'$ and $p''$ lie in the cone that is bounded by the rays $\gamma_i + q_0$ and $\gamma_{i+1} + q_0$, we can deduce from the property of $d(\xi_i,\xi_{i+1}) \leq 1/4$ that $\{p',p''\} \subset \frac{\lambda_0}{4}\hat{Q} + \lambda_0 \xi_i + q_0$.  Therefore, $d(p',p'') \leq \lambda_0/2$.  The edge of $\lambda_0\hat{Q}+q_0$ that contains $p''$ and the edge of  $\lambda_1\hat{Q}+q_0$ that contains $p$ are homothetic copies of the same edge of $\hat{Q}$.  Since $p$, $p''$ and $q_0$ are collinear, the Euclidean length of $pp''$ is $(\lambda_1-\lambda_0)/\lambda_0$ times the Euclidean length of $q_0p''$.  Therefore, $(\lambda_1-\lambda_0)\hat{Q} + p''$ contains $p$ in its boundary, which implies that $d(p,p'') = \lambda_1 - \lambda_0$.  By the triangle inequality,
		\begin{align}
			d(p,p') &\leq d(p,p'') + d(p',p'') \leq \lambda_1 - \lambda_0/2 \nonumber \\
			              & = d(p,q_0) - d(p',q_0)/2.  \label{eq:0}
		\end{align}
		As mentioned at the beginning of this proof, since $p' \in C$, there exists $\{P_w,P_{w'}\} \in \Delta$ such that $p' \in P_{w'}$, and $w$ is $u$ or an ancestor of $u$.  So $q_0 \in P_u \subseteq P_w$.  Let $\delta_u$ and $\delta_w$ be the diameters of $P_u$ and $P_w$ under $d$, respectively.  We have $d(p',q_0) \geq d(p',P_w) \geq d(P_{w'},P_w)$.  Since $\Delta$ is a 4-WSPD, $d(P_{w'},P_w) \geq 4\delta_w \geq 4\delta_u$.  It follows that $d(p',q_0)/2 \geq 2\delta_u$.  Substituting into \eqref{eq:0} gives $d(p,p') \leq d(p,q_0) - 2\delta_u$.
		
		For any point $y \in P_u$, $d(p,p') \leq d(p,q_0) - 2\delta_u \leq d(p,y) + d(q_0,y) - 2\delta_u \leq d(p,y) - \delta_u < d(p,y)$.  As a result, $p$ is closer to $p'$ than any point in $P_u$.  If there are at least $k$ such $p'$'s, no point in $P_u$ can be a $k$-nearest neighbor of $p$.
\end{proof}

\subsection{Nearest neighbor graph}
\label{app:NN-3}
	
We restate Lemma~\ref{lem:deg} and give its proof.

\vspace{8pt}

\noindent {\bfseries\sffamily Statement of Lemma~\ref{lem:deg}:}\hspace{4pt}\emph{For any subset $X \subseteq S$, every vertex in $1$-$\NN_X$ has $O(1)$ degree, and adjacent vertices in $1$-$\NN_X$ are Voronoi neighbors in $\vor(X)$.}
\begin{proof}
	For every point $p \in X$, let $\hat{Q}_p$ be the largest homothetic copy of $\hat{Q}$ centered at $p$ such that $\Int(\hat{Q}_p) \cap (X \setminus \{p\}) = \emptyset$.  In 1-$\NN_X$, a point $q \in X$ is connected to its nearest neighbor, and if any other point $p \in X$ is connected to $q$, then $q \in \partial \hat{Q}_p$.   Therefore, the vertex degree of 1-$\NN_X$ is bounded from above by the maximum number of polygons in $\{\hat{Q}_p : p \in X\}$ that are intersected by a point in $\mathbb{R}^2$.

		Let $y$ be a point in $\mathbb{R}^2$ that intersects the maximum number of polygons in $\{\hat{Q}_p : p \in X\}$.  Let $\{\hat{Q}_{p_1}, \ldots, \hat{Q}_{p_s}\}$ be the polygons intersected by $y$.  Shrink each $\hat{Q}_{p_i}$ concentrically to a polygon $\hat{Q}_i$ that just contains $y$ in its boundary.  So $\hat{Q}_i \subseteq \hat{Q}_{p_i}$, meaning that $\Int(\hat{Q}_i) \cap (X \setminus \{p_i\}) = \emptyset$.  Take the largest $\lambda > 0$ such that the interior of $\hat{Q}_y = \lambda\hat{Q} + y$ does not intersect $\{p_1, \ldots, p_s\}$.  For $i \in [1,s]$, let $q_i$ be the intersection between the segment $p_iy$ and $\partial \hat{Q}_y$.  Refer to Figure~\ref{fg:deg}.
		
		\begin{figure}
			\centerline{\includegraphics[scale=0.6]{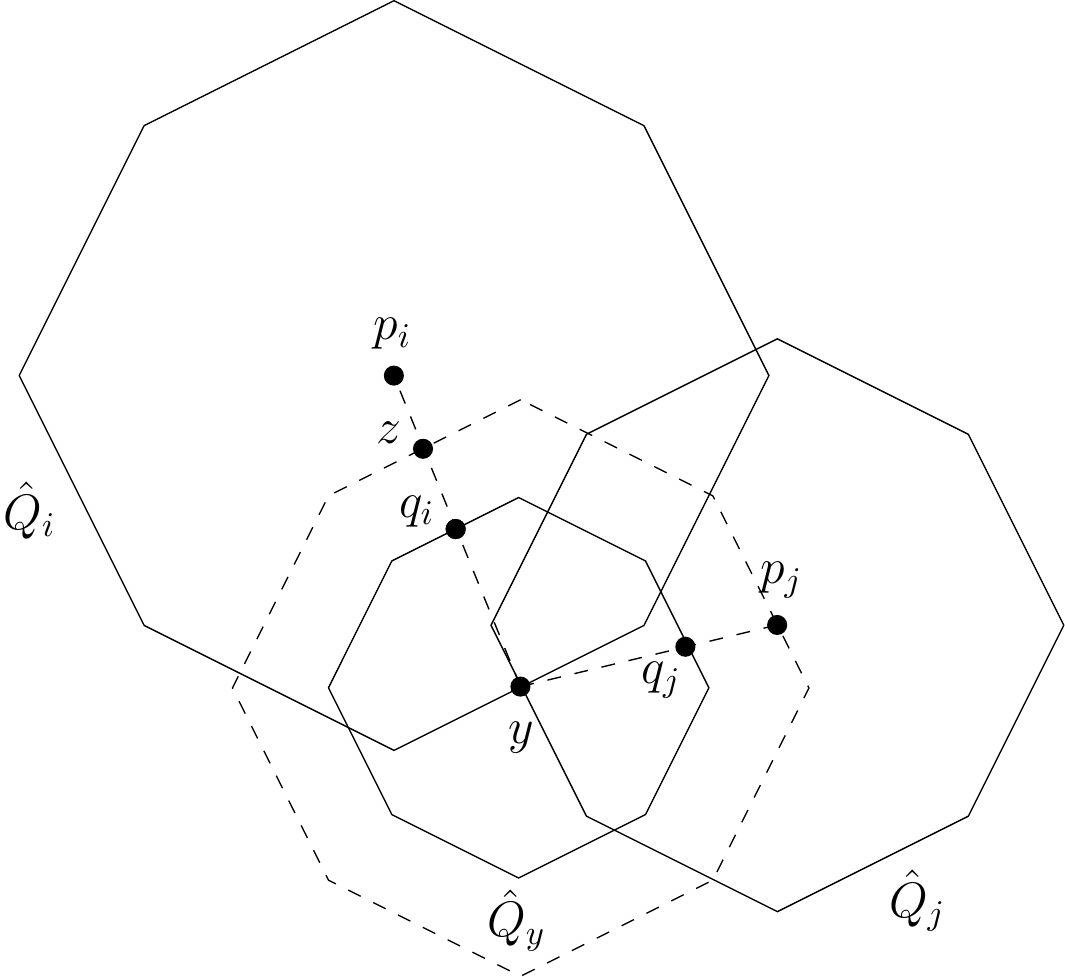}}
			\caption{Illustration for the proof of Lemma~\ref{lem:deg}.}
			\label{fg:deg}
		\end{figure}
		
		We claim that $d(q_i,q_j) \geq \lambda$ for all $i \not= j$.  Without loss of generality, assume that $d(y,p_i) \geq d(y,p_j)$.  Let $z$ be the point in the segment $p_iy$ such that $d(y,z) = d(y,p_j)$.  Since $y$, $z$ and $p_i$ are collinear, we have $d(y,z) = d(y,p_i) - d(p_i,z)$.   Since $p_j \not\in \Int(\hat{Q}_{p_i})$ and $y \in \hat{Q}_{p_i}$, we have $d(p_i,p_j) \geq d(y,p_i)$, which implies that $d(y,z) \leq d(p_i,p_j) - d(p_i,z) \leq d(p_i,z) + d(p_j,z) - d(p_i,z) = d(p_j,z)$.    Since the wedge $yq_iq_j$ is a scaled copy of the wedge $yzp_j$, the inequality $d(p_j,z) \geq d(y,z)$ implies that $d(q_i,q_j) \geq d(y,q_i) = \lambda$.
		
		Our claim implies that we can place non-overlapping copies of $\frac{\lambda}{2}\hat{Q}$ centered at the $q_i$'s.  Each copy has half the area of $\hat{Q}_y$, and all these copies are contained inside $2\hat{Q}_y$.  A packing argument shows that there are $O(1)$ such copies.  This shows that the vertex degree of 1-$\NN_X$ is $O(1)$.
		
		Let $pq$ be an edge in 1-$\NN_X$.  We assume without loss of generality that $p$ is the nearest neighbor of $q$ in $X$.  Thus, there exists $\lambda > 0$ such that $p \in \partial(\lambda\hat{Q}+q)$ and $\Int(\lambda\hat{Q}+q) \cap X = \emptyset$.  By the definition of $\hat{Q}$, it means that there exists a point $x \in \mathbb{R}^2$ such that $\lambda Q^* + x \subset \lambda\hat{Q} + q$ and $\{p,q\} \subset \partial(\lambda Q^* + x)$.  So  $\Int(\lambda Q^* + x) \cap X = \emptyset$ which certifies that $p$ and $q$ are Voronoi neighbors in $\vor(X)$.
\end{proof}
		
\section{Missing details in Section~\ref{sec:vorNN}}
\label{app:vorNN}

			Only the proof of Lemma~\ref{lem:delNN} is missing.  We restate the lemma and give its proof below.  The proof contains the details of the analysis in~\cite{BM11} that works for step~6 of VorNN.
			
			\vspace{12pt}
		
\noindent {{\bfseries\sffamily Statement of Lemma~\ref{lem:delNN}:}\hspace{4pt}\emph{\emph{VorNN}$(R,T_R)$ computes $\vor(R)$ in $O(|R|)$ expected time.}
\begin{proof}
			We first give the details of step~6 of VorNN.   As in~\cite{BM11}, we grow $X$ and $\vor(X)$ by moving points repeatedly from $Y \setminus X$ to $X$.   We keep track of edges $pq$ in 1-$\NN_Y$ such that $p \in Y \setminus X$ and $q \in X$.  Take such an edge.
			
			Since $pq$ is an edge in 1-$\NN_Y$, $pq$ must also be an edge in 1-$\NN_{X \cup \{p\}}$.  By Lemma~\ref{lem:deg}, $p$ and $q$ are Voronoi neighbors in $\vor(X \cup \{p\})$.   So $p$ must conflict with a point in $V_q(X)$.  Lemma~\ref{lem:conflict-tech} implies that $p$ must conflict with some Voronoi edge bend or Voronoi vertex $v$ in $\partial V_q(X)$.  We search $\vor(X)$ from $v$ to find all Voronoi edge bends and Voronoi vertices that conflict with $p$.  During this search, we can modify $\vor(X)$ to $\vor(X \cup \{p\})$ in time proportional to the number of Voronoi edge bends and Voronoi vertices that conflict with $p$~\cite{KMM93}.
			
			
			Consider the total running time.   The identification of the starting Voronoi edge bend or Voronoi vertex involves checking $\partial V_q(X)$.  In other words, $\partial V_q(X)$ is examined once for each neighbor of $q$ in 1-$\NN_Y$.  Therefore, any Voronoi edge bend or Voronoi vertex $w$ that was constructed at some point during the algorithm can be examined as many times as the degree sum in 1-$\NN_Y$ of the points that define $w$.  This degree sum is $O(1)$ by Lemma~\ref{lem:deg}.  Also, if $w$ is subsequently destroyed by the insertion of a point in $Y \setminus X$, we can charge the time to destroy $w$ to the creation of $w$.  Hence, the total expected running time is bounded by the expected number of Voronoi edge bends and Voronoi vertices created in the course of the algorithm.
			
			Let $Q^*_w$ denote the homothetic copy of $Q^*$ that circumscribes the points that define $w$.  Let $s = |Y \cap Q^*_w|$.  A necessary condition for $w$ to be created in step~6 is that $Y \cap Q^*_w$ contains no point in $X$ right before the execution of step~6.  If some points in $Y \cap Q^*_w$ form matching pair(s) in step~3 of VorNN, at least one of them must be included in $X$ in step~3, which means that we cannot possibly create $w$ in step~6.  If the points in $Y \cap Q^*_w$ do not form any matching pair in step~3, then the points in $Y \cap Q^*_w$ are sampled independently with probability 1/2.   Therefore, the probability that none of these points is selected in step~3 is at most $1/2^s$, so the probability that $w$ is created in step~6 is at most $1/2^s$.  By the result of Clarkson and Shor~\cite[Theorem~3.1]{CS89}, there are $O(|Y|s^2)$ Voronoi edge bends and Voronoi vertices whose circumscribing homothetic copies of $Q^*$ contain at most $s$ points in $Y$.   Therefore, the expected number of Voronoi edge bends and Voronoi vertices created in step~6 of VorNN is $O(\sum_{s = 0}^\infty |Y|s^2/2^s) = O(|Y|)$.  The expected size of $X$ is $|Y|/2$.   Therefore, unwinding the recursion starting from the top-level call VorNN$(R,T_R)$ gives a total expected running time of $O(|R| + |R|/2 + |R|/4 + \cdots) = O(|R|)$.
\end{proof}

\section{Missing details in Section~\ref{sec:I}}
\label{app:I}

Recall that $U_R$ is the set of Voronoi edge bends and Voronoi vertices in $\vor(R)$ that conflict with the input points $p_1,\ldots, p_n$.  We need to prove that $U_R = \bigcup_{i=1}^n V_S|_{p_i}$.  Clearly, $\bigcup_{i=1}^n V_S|_{p_i} \subseteq U_R$ by the definition of $R$.  The following result proves that the containment also holds in the other direction.

\begin{lemma}
	$U_R \subseteq \bigcup_{i=1}^n V_S|_{p_i}$.
\end{lemma}
\begin{proof}
	Take any $p_i \in I$.  For each $q \in N_{p_i}\bigl(S \cup \{p_i\}\bigr)$, $p_i$ must conflict with $V_q(S)$ in order that $q$ becomes a Voronoi neighbor of $p_i$ in $\vor\big(S \cup \{p_i\}\bigr)$.  As a result, $N_{p_i}\big(S \cup \{p_i\}\bigr) \subseteq R$, which implies that  $V_{p_i}\bigl(S \cup \{p_i\}\bigr) = V_{p_i}\bigl(R \cup \{p_i\}\big)$.  In $\vor(S)$,  the region $V_{p_i}\bigl(S \cup \{p_i\}\bigr)$ is partitioned and distributed among the Voronoi cells of points in $N_{p_i}\bigl(S \cup \{p_i\}\bigr) \subseteq R$.  Thus, any Voronoi edge bend or Voronoi vertex $w$ in $\vor(R)$ that does not exist in $V_S$ must lie strictly outside $V_{p_i}\bigl(S \cup \{p_i\}\bigr) = V_{p_i}\bigl(R \cup \{p_i\}\big)$.  Hence, $w$ cannot conflict with $p_i$.	  
\cancel{
	Assume to the contrary that there is some Voronoi edge bend or Voronoi vertex $v$ in $\vor(R)$ that conflicts with some $p_i \in I$, but $v \not\in V_S$.  Since $v \not\in V_S$, there is a point $q \in S \setminus R$ such that $v \in V_q(S)$.    As $q \not\in R$, $p_i \not\in V_q(S)$.  Therefore, $p_iq$ intersects $\partial V_q(S)$ at a point $x$.  If $d_Q(p_i,x) > d_Q(q,x)$, there must be a point $y \in \Int(qx) \subseteq V_q(S)$ such that $d_Q(p_i,y) = d_Q(q,y)$.  But then $p_i$ conflicts with $V_q(S)$, contradicting the fact that $q \not\in R$.  Suppose that $d_Q(p_i,x) \leq d_Q(q,x)$.   Let $Q^*_x$ be the largest homothetic copy of $Q^*$ centered at $x$ such that $\Int(Q^*_x) \cap S = \emptyset$.  We have $q \in \partial Q^*_x$ as $x \in V_q(S)$.  It means that $p_i \in Q^*_x$ as $d_{Q*}(x,p_i) = d_Q(p_i,x) \leq d_Q(q,x) = d_{Q^*}(x,q)$.  By Lemma~\ref{lem:enclose}, there is a Voronoi edge bend or Voronoi vertex $w$ in $\partial V_q(S)$ such that $p \in Q^*_w$, where $Q^*_w$ is the largest homothetic copy of $Q^*$ centered at $w$ such that $\Int(Q^*_w) \cap S = \emptyset$.  But then $p_i$ conflicts with $V_q(S)$, a contradiction to the fact that $q \not\in R$.
}
\end{proof}


\begin{thebibliography}{10}
	
\bibitem{arya07}
S. Arya, T. Malamatos, D.M. Mount, and K.C. Wong.
Optimal expected-case planar point location.
\emph{SIAM Journal on Computing}, 37~(2007), 584--610.

\bibitem{ailon11}
N.~Ailon, B.~Chazelle, K.~Clarkson, D.~Liu, W.~Mulzer, and C.~Seshadhri.
\newblock Self-improving algorithms.
\newblock {\em SIAM Journal on Computing}, 40~(2011), 350--375.

\bibitem{BM11}
K.~Buchin and W.~Mulzer.
\newblock Delaunay triangulations in $o(\text{sort}(n))$ time and more.
\newblock {\em Journal of the ACM}, 58~(2011), 6:1--6:27.

\bibitem{CK95}
P.B.~Callahan and S.R.~Kosaraju.
\newblock A decomposition of multidimensional point sets with applications to
$k$-nearest-neighbors and $n$-body potential fields.
\newblock {\em Journal of the ACM}, 42~(1995), 67--90.

\bibitem{chan16}
T.M.~Chan.
A simpler linear-time algorithm for intersecting two convex polyhedra in three dimensions.
{\em Discrete {\&} Computational Geometry}, 56~(2016), 860--865.



\bibitem{chazelle02}
B.~Chazelle, O.~Devillers, F.~Hurtado, M.~Mora, V.~Sacristan, and M.~Teillaud.
\newblock Splitting a delaunay triangulation in linear time.
\newblock {\em Algorithmica}, 34~(2002), 39--46.

\bibitem{cheng20a}
S.-W.~Cheng, M.-K.~Chiu, K.~Jin, and M.T.~Wong.
A generalization of self-improving algorithms.
\emph{Proceedings of the International Symposium on Computational Geometry}, 2020, 29:1--29:13.  Full version: arXiv:2003.08329v2.


\bibitem{cheng20b}
S.-W.~Cheng, K.~Jin, and L.~Yan.
\newblock Extensions of self-improving sorters.
\newblock \emph{Algorithmica}, 82~(2020), 88--106.


\bibitem{CD85}
L. Paul Chew and R.L. Scot Drysdale.
Voronoi diagrams based on convex distance functions.
\emph{Proceedings of the 1st Annual Symposium on Computational Geometry}, 1985, 235--244.



\bibitem{clarkson14}
K.L.~Clarkson, W.~Mulzer, and C.~Seshadhri.
\newblock {Self-improving algorithms for coordinatewise maxima and convex
	hulls}.
\newblock {\em SIAM Journal on Computing}, 43(2):617--653, 2014.


\bibitem{CS89}
K.L. Clarkson and P.W. Shor.
Applications of random sampling in computational geometry,~II.
\emph{Discrete and Computational Geometry}, 4~(1989), 387--421.

\bibitem{cover06}
T.M.~Cover and J.A.~Thomas.
\newblock {\em {Elements of Information Theory}}.
\newblock Wiley-Interscience, New York, 2nd edition, 2006.




\bibitem{edels86}
H. Edelsbrunner, L.J. Guibas, and J. Stolfi.
Optimal point location in a monotone subdivision.
\emph{SIAM Journal on Computing}, 15~(1986), 317--340.



\bibitem{F87}
G.N. Frederickson.
Fast algorithms for shortest paths in planar graphs, with applications.
\emph{SIAM Journal on Computing}, 16~(1987), 1004--1022.


\bibitem{FREDMAN76}
M.L.~Fredman.
\newblock How good is the information theory bound in sorting?
\newblock {\em Theoretical Computer Science}, 1(4):355 -- 361, 1976.


\bibitem{har-peled06}
S. Har-Peled and M. Mendel.
Fast construction of nets in low-dimensional metrics and their applications.
\emph{SIAM Journal on Computing}, 35~(2006), 1148--1184.


\bibitem{iacono04}
J.~Iacono.
\newblock Expected asymptotically optimal planar point location.
\newblock {\em Computational Geometry: Theory and Applications}, 29 (2004), 19--22.

\bibitem{JP18}
K. Junginer and E. Papadopoulou.
Deletion in abstract Voronoi diagram in expected linear time.
\emph{Proceedings of the 34th International Symposium on Computational Geometry}, 2018, 50:1--50:14.

\bibitem{KMM93}
R. Klein, K. Mehlhorn, and S. Meiser.
Randomized incremental construction of abstract Voronoi diagrams.
\emph{Computational Geometry: Theory and Applications}, 3~(1993), 157--184.






\bibitem{miller97}
G.L. Miller, S.-H. Teng, W. Thurston, and S.A. Vavasis.
Separators for sphere-packings and nearest neighbor graphs.
\emph{Journal of the ACM}, 44~(1997), 1--29.

\bibitem{PR08}
E. Pyrga and S. Ray.
New existence proofs for $\epsilon$-nets.
\emph{Proceedings of the 24th Annual Symposium on Computational Geometry}, 2008, 199--207.


\bibitem{TL88}
A.K. Tsakalides and J. van Leeuwen.
An optimal pointer machine algorithm for finding nearest common ancestors.
Technical Report RUU-CS-88-17, Department of Computer Science, University of Utrecht, 1988.

\bibitem{boas77}
P. van Emde Boas, R. Kaas, and E. Zijlstra.
Design and implementation of an efficient priority queue.
\emph{Mathematical Systems Theory}, 10~(1977), 99--127.

\bibitem{ray}
R.W. Yeung.
\emph{A First Course in Information Theory},
Kluwer Academic/Plenum Publishers, 2002.

\end{thebibliography}
\end{document}